%% file: utility-max.tex
\documentclass[]{article}
\usepackage{booktabs}
 \usepackage{enumitem}
\usepackage{authblk}
\usepackage{fullpage}
\usepackage[vlined,ruled,noend]{algorithm2e}
\usepackage[noend]{algorithmic}
\usepackage{graphicx,dsfont,bm,subfig}
\usepackage{times}
\usepackage{mathtools}
\usepackage{color}

\newcommand*{\xhdr}[1]{\vspace{0.5mm} \noindent{{\bf #1.}}}
\usepackage[english]{babel}
 \usepackage{latexsym}
\usepackage{hyperref}
\usepackage{mathrsfs}
\usepackage{bbm}
\usepackage{etoolbox}
\usepackage{Definitions}
\usepackage{cleveref}
\usepackage[]{algorithm2e}
\graphicspath{{./FIG/}}

\title{Can A User Anticipate What Her Followers Want?}
\date{}
\author{Abir De}
\author{Adish Singla}
\author{Utkarsh Upadhyay}
\author{Manuel Gomez-Rodriguez}
\affil{Max Planck Institute for Software Systems, \\ \{ade,adishs,utkarshu,manuelgr\}@mpi-sws.org}

\begin{document}
\maketitle

\begin{abstract}
\input{000abstract.tex}
\end{abstract}
\maketitle

\vspace{-1mm}
\section{Introduction}
\label{sec:introduction}
\input{010introduction.tex}


\vspace{-1mm}
\section{Feedback~model~of~posting~behavior}
\label{sec:model}
\input{020model.tex}
\vspace{-1mm}
\section{Solving the utility maximization problem}
\label{sec:properties}
\input{030utility.tex}
\vspace{-5mm}
\section{Utility estimation framework}
\label{sec:estimation}
\vspace{-2mm}
\input{031estimation.tex}
\section{Experiments on synthetic data}
\label{sec:experiments}
\input{040syn-experiments.tex}
\vspace{-3mm}
\section{Experiments on real data}
\label{sec:experiments}
\input{041real-experiments.tex}
 \vspace{-3mm}
\section{Conclusions}
\label{sec:conclusions}
 \vspace{-1mm}
\input{050conclusion.tex}
%

%

{
\small
\bibliographystyle{abbrv}
\bibliography{refs}
}
\clearpage
\newpage
\input{appen.tex}


\end{document}

%% file: 000abstract.tex
Whenever a social media user decides to share a story, she is ty\-pi\-cally pleased to receive likes, comments, shares, or,
more generally, \emph{feedback} from her follo\-wers.
%
%
As a result, she may feel compelled to use the feedback she receives to (re-)estimate her follo\-wers'{} pre\-fe\-ren\-ces and 
decide which stories to share next to receive more (positive) feedback.
Under which conditions can she \emph{succeed}?
In this work, we first look into  this problem from a theoretical pers\-pec\-tive and then provide a set of practical algorithms to 
identify and characterize such behavior in social media.
%

More specifically, we address the above problem from the perspective of sequential decision making and utility maximization. 
For a wide variety of utility functions, we first show that, to succeed, a user requires to actively trade off exploitation---sharing stories 
which lead to more (positive) feedback---and exploration---sharing stories to learn about her followers'{} preferences.
%
However, exploration is not necessary if a user utilizes the feedback her followers provide to other users in addition to the feedback
she receives.
Then, we develop a utility estimation framework for observation data, which relies on statistical hypothesis testing to determine whether 
a user utilizes the feedback she receives from each of her followers to decide what to post next.
%
%
Experiments on synthetic data illustrate our theoretical fin\-dings and show that our estimation framework is able to accurately recover 
users'{} underlying utility functions.
%
Experiments on several real datasets gathered from Twitter and Reddit reveal that up to $82$\% ($43$\%) of the Twitter (Reddit) users in our 
datasets do use the feedback they receive to decide what to post next.

%% file: 010introduction.tex
%

Political parties, corporations, celebrities as well as ordinary people use social media to build, reach, and share stories with their own audience.
For example, political leaders share details about their activities in hopes of tapping new voters~\cite{dailyiowan}, 
corporations offer insights about their latest products and services with potential customers~\cite{de2012popularity},
celebrities give a glimpse of their lavish lifestyle to strengthen their fan base~\cite{wsj}, 
and
ordinary people share personal stories with their friends~\cite{huffingtonpost}. 
In all these cases, social media users---politicians, corporations, celebrities, or ordinary people---receive feedback from their followers---their voters, customers, fans, or friends---by means of 
likes, comments, or shares.
Moreover, this feedback provides hints about the preferences of their followers: it lets the users know what does or does not work, and it influences what they share next, as shown by an 
increasing number of em\-pi\-ri\-cal studies~\cite{grinberg2016changes,grinberg2017understanding,dholakia2004social,joinson2008looking,cheng2014community,burke2010social,burke2013using,burke2014growing,ellison2014cultivating,gray2013wants}.
In this context, it is perhaps surprising that feedback models of posting behavior are largely nonexistent to date. However, such models are of outstanding interest since they 
would allow us to answer two fundamental questions:
\vspace*{-1mm}
\begin{itemize}[leftmargin=0.7cm]
%
\item[(i)] Can a user succeed at maximizing the (positive) feedback she receives if, \emph{a priori}, does not know her 
followers'{} preferences?

\item[(ii)] Can we determine whether a user utilizes the feedback she receives from each of her followers to decide what
to post next using observational data? 
\end{itemize}
%
%
%
%
%
%
%

%
%
%
By answering the above questions, we will not only advance our understanding of how feedback may influence a user'{}s posting 
behavior but will also facilitate the design of more effective algorithms for viral marketing and user personalization.

\vspace{-2mm}
\subsection{Overview of our Approach}
\vspace{-1mm}
In this paper, we introduce a utility maximization feedback model of posting behavior, which is specially well-fitted to investigate the above
questions. 
%
%
More specifically,
we assume that each user has an underlying (linear) utility function, which assigns different weights to the feedback the user receives from each of her followers.
%
Moreover, every time the user shares a story with her followers, they provide their feedback according to a set of preferences.
%
%
If the user knows perfectly her followers'{} preferences, we can then show that the optimal posting strategy, which maximizes the 
user'{}s utility function, is deterministic.
If the user does not know their preferences, we have the following theoretical results:
\begin{itemize}[leftmargin=0.7cm]
\item[I.] If the user estimates her followers'{} preferences from the feedback she receives over time,
%
%
she needs to resort to posting strategies that effectively trade-off exploitation---sharing stories to maximize her utility---and exploration---sharing
stories to learn about her followers'{} preferences.
%
More formally, we can show that \emph{posterior sampling} based posting strategies achieve logarithmic regret (\ie, $O(\log T)$) while 
strategies based on \emph{point estimates} suffer from linear regret (\ie, $\Theta(T)$) where $T$ denotes the total time steps.

\item[II.] If the user can, in addition to the feedback she receives, also use the feedback her followers give to other users to estimate
her followers'{} preferences, she is better off using it.
%
More specifically, we can show that posterior sampling based posting strategies achieve constant regret (\ie, $O(1)$) and, perhaps surprisingly, 
strategies based on point estimates achieve sublinear regret (\ie, $o(T)$).
\end{itemize}
%
%
%
In addition to the above theoretical analysis, we also develop a utility estimation framework, which relies on statistical hypothesis testing to determine whether 
a user utilizes the feedback she receives from each of her followers to decide what to post next. 
%
%
%
Finally, we perform a variety of experiments using both synthetic and real Twitter and Reddit data.
Experiments on synthetic data illustrate our theoretical findings and show that our utility estimation framework is able to accurately recover the users'{} underlying 
utility functions.
%
%
%
Experiments on several real datasets gathered from Twitter and Reddit reveal that up to $82$\% ($43$\%) of the Twitter (Reddit) users in our 
datasets do use the feedback they receive to decide what to post next.
\vspace{-3mm}
\subsection{Related Work}
\vspace{-1mm}
In addition to the empirical studies on how feedback influences user behavior, discussed previously, our work also relates to revealed preference
theory, smart broadcasting, and multi-armed bandits.
%

\xhdr{Revealed preference theory}
%
Since the pioneering work by Samuelson~\cite{samuelson1948consumption}, reveal preference theory has become a well-established economic theory that
analyzes choices made by individuals, particularly to understand consumer behavior.
It typically assumes that each consumer decides to buy a bundle of goods, among several alternatives, on the basis of a (concave) nondecreasing utility
function~\cite{koo1963empirical,mas1977recoverability,afriat1967construction}.
The works most closely related to ours~\cite{beigman2006learning,zadimoghaddam2012efficiently,echenique2011revealed,balcan2014learning,uchange} aim to
develop efficient algorithms to estimate utility functions from revealed preference data as well as analyze their sample complexity.
However, their problem setting is very different from ours:
(i) the utility a consumer obtains from buying a bundle of goods is deterministic, however, the feedback a social media user receives from her followers varies
randomly and, thus, the utility a user obtains from sharing a story is stochastic;
(ii) a consumer can evaluate the utility of a bundle of goods exactly, however, a social media user needs to guess the utility she will obtain from sharing a story
on the basis of an estimation of her followers'{} preferences from the feedback she received in the past;
and, (iii) each consumer'{}s decision is independent, however, each social media user'{}s decision is part of a sequential decision making process.

\xhdr{Smart broadcasting}
In recent years, there has been some work on smart broadcasting~\cite{spasojevic2015post, redqueen17wsdm, smart16, jmlr}, which aims to find the times when a user should post
to receive more views, likes, comments, or shares from her followers.
In smart broadcasting, there is also a user who aims to maximize the impact of the stories she shares, however, the focus is on \emph{when} to share while our focus is on \emph{what} to share.
In addition, algorithms for smart broadcasting are based on temporal point processes and stochastic optimal control, while we resort to online learning techniques and convex optimization.


\xhdr{Multi-armed bandits}
The proof techniques used for deriving regret bounds in bandit problems~\cite{tspaper,auer2007logarithmic,jaksch2010near,agrawal2013thompson} are related to the ones we use to derive the
regret bounds in our work.
However, there are several key differences:
(i) we allow a user to utilize the feedback her followers provide to other users to estimate her followers'{} preferences, whereas, in a traditional bandit setting, a user could only utilize the feedback
she receives on the stories she shares;
(ii) our proof techniques allow for users with several followers, in contrast, in a traditional bandit setting, it would only allow for users with one follower.

%% file: 020model.tex
%
In this section, we introduce our feedback model of posting behavior, starting from the problem setting it is designed for. 
%

%
%

%
%
\vspace{-1mm}
\subsection{Problem Setting}
\vspace{-1mm}
Let $u$ be a social media user and $\Ncal(u)$ be her set of followers. Then, at each time $t \in \{1, \ldots, T\}$, the user shares
a story from a topic $c_t \in \Ccal$, with $|\Ccal| = K$,
and each of her followers decides whether to give (or not to give) feedback.\footnote{For simplicity, we assume that receiving
feedback is always positive, \eg, we assume that receiving a like is always positive. Our formulation can
be easily adapted to scenarios with positive and negative feedback, \eg, upvotes and downvotes.}
%
Moreover, for each follower $v \in \Ncal(u)$, we denote the feedback she gave (or did not give) to user $u$ and, possibly, to any other user in the social media
platform\footnote{We will study two scenarios: (i) the user has access only to the feedback she has received; and (ii) the user has access to the feedback she has
received as well as the feedback her followers have given to other users.}, up to time $t-1$ by
$$\Hcal_{v}(t) = \{ (c, l(c)) \, | \, \mbox{the story was posted before} \, t \},$$
where $c$ denotes the topic of the story, $l(c) = 1$ means that she gave feedback to the story and $l(c)= 0$ otherwise. We denote the collection of  feedback from followers $\Ncal(u)$ as $\Hcal(t)= (\Hcal_{v}(t))_{v \in \Ncal(u)}$.

At time $t \in \{1, \ldots, T\}$, we assume that the user $u$ samples the topic $c$ of the story she shares from a categorical distribution $c_t \sim p(c | \Hcal(t))$, which may depend on the history of her followers'{} feedback.
Also, given a topic $c$, each follower $v$ gives feedback with conditional probability $p_{v}(l(c) = 1 | c) = q_{cv}$, where we can think of $q_{cv}$ as follower
$v$'{}s \emph{preference} for topic $c$.
Here, we assume that, in general, followers may differ in their preferences, \ie, $q_{cv} \neq q_{cv'}$ for $v \neq v'$.
Moreover, we denote $\qb_c=(q_{cv})_{ v\in\Ncal(u) }$ and $\Qb=(q_{cv})_{ c \in\Ccal, v\in\Ncal(u)}$.

%
%
%
%
%
%
%
%
\vspace{-1mm}
\subsection{Utility Maximization Feedback Model}
\vspace{-1mm}
We assume that user $u$ aims to find (and utilize) the categorical distribution $p^{*}(c | \Hcal(t))$ that maximizes a (linear) utility 
function $\us(T)$, defined as
%
\begin{align}
\us(T) = \EE \left[ \sum_{t\in[T]} \left(\sum_{v \in \Ncal(u)} a_v l_v(c_t)+a_u x_{c_t}\right) \right] \label{eq:basic0}
\end{align}
where
%
the expectation is over the topics $c_t \sim p^{*}(c | \Hcal(t))$ of the stories the user shares and the feedback
$l_v(c_t) \sim \text{Bernoulli}(q_{c_t v})$.
The weights $a_v \geq 0$ model the importance that user $u$ gives to the feedback she receives from
follower $v$,
the parameters $0 \leq x_{c_t} \leq 1$ encode user $u$'s preference for topic $c_t$,
the weight $a_u \geq 0$ models the importance she gives to her own preferences,
and we assume that $\sum_{v \in \Ncal(u)} a_v + a_u = 1$.
In the above utility function, the greater the importance $a_v$ user $u$ gives to the feedback she receives from a specific
follower $v$, the greater the utility gain she will obtain from posting a story from a topic that follower $v$ prefers.
Similarly, the greater the importance $a_u$ she gives to her own preferences $\xb = (x_{c})_{c \in \Ccal}$, the greater the utility gain she will obtain from
just posting a story from a topic she prefers.

Next, using the linearity of expectation and the law of iterated expectation, we can rewrite Eq.~\ref{eq:basic0} as
follows:
%
%
\begin{align}
\us(T) &=  \sum_{t\in[T]} \sum_{v \in \Ncal(u)} \EE_{c_t \sim p^{*}}\left[ \EE_{l_v(c_t) | c_t}\left[ a_v l_v(c_t)+a_u x_{c_t} \right]\right]   \nonumber \\
&=\sum_{t\in[T]} \sum_{v \in \Ncal(u)} \EE_{c_t \sim p^{*}}\left[ a_v q_{c_t v}+a_u x_{c_t} \right] \nonumber \\
&= \sum_{t\in[T]} \sum_{c \in \Ccal} p^{*}(c | \Hcal(t)) \left( \ab^{\dagger} \qb_c + a_u x_{c} \right) \label{problem:basic}
\end{align}
where $\ab = (a_v)_{v \in \Ncal(u)}$ and $\dagger$ denotes the transpose operator.

Finally, we can formally state the utility maximization problem user $u$ aims to solve as:
\begin{equation} \label{eq:utility-maximization}
\begin{aligned}
\underset{p(c | \Hcal(t))}{\text{maximize}} & \quad \sum_{t\in[T]} \sum_{c \in \Ccal} p(c | \Hcal(t)) \left( \ab^{\dagger} \qb_c + a_u x_{c} \right) \\
\text{subject to} & \quad 0 \leq p(c | \Hcal(t)) \leq 1 \quad \forall c \in \Ccal,\\
& \sum_{c \in \Ccal} p(c | \Hcal(t)) = 1,
\end{aligned}
\end{equation}
where the followers'{} preferences $\Qb = (\qb_c)_{c \in \Ccal}$ are generally unknown. 

%% file: 030utility.tex
In this section, we develop an efficient algorithm to solve the utility maximization problem defined by Eq.~\ref{eq:utility-maximization}
and study its theoretical guarantees in a variety of settings.
%

\vspace{-1mm}
\subsection{Known Preferences}
\vspace{-1mm}
As a warm up, we first assume that the user $u$ knows her followers'{} preferences $\Qb$.
%
In this setting, it readily follows that the optimal distribution $p^{*}(c | \Hcal(t))$ that
maximizes Eq.~\ref{eq:utility-maximization} does not depend on the feedback history and is given by
\begin{equation}\label{eq:popt}
p^{*}(c | \Hcal(t)) = p^{*}(c) = \II\left[ c = \argmax_{c'} \left( \ab^{\dagger} \qb _{c'}+ a_u x_{c'}) \right) \right],
\end{equation}
where $\II(\cdot)$ denotes the indicator function. In this case, note that the optimal mechanism for maximizing utility is deterministic (assuming ties over different topics are broken in a predetermined way).

%

\vspace{-1mm}
\subsection{Unknown Preferences: Exploration Exploitation Trade-off}
\vspace{-1mm}
%

In this section, we consider a more realistic setting where the user $u$ does not know the preferences $\Qb$. For this setting, we assume that the user can access historical feedback data to estimate these unknown preferences. When user $u$ does not know her followers'{} preferences, she needs to trade off exploitation, \ie, maximizing utility, and exploration, \ie, learning about
her followers'{} preferences $\Qb$ from historical feedback data.

To this aim, for every topic $c \in \Ccal$ and follower $v \in \Ncal(u)$, we assume a Beta prior over the preference parameter $q_{cv} \sim Beta(\alpha,\beta)$.
Under this assumption, at each time $t$, we can use $\Hcal_v(t)$ to update the distribution of parameter $q_{cv}(t)$ as:
\begin{equation}
p(q_{cv}(t) | \Hcal_v(t)) = Beta(\alpha+n_{cv}(t),\beta+\bar{n}_{cv}(t)),
\end{equation}
where
\begin{align*}
 n_{cv}(t) &= |\{ (c', l(c')) \in\Hcal_v(t) \,|\, l(c')=1, c'=c\}|, \\
 \bar{n}_{cv}(t) &= |\{ (c', l(c')) \in\Hcal_v(t) \,|\, l(c')=0, c'=c \}|.
\end{align*}

Then, at the beginning of time $t$, we can estimate the value of each preference parameter $q_{cv}(t)$ in the following two ways:
\begin{itemize}[leftmargin=0.7cm]
\item[I.] using point estimates, \ie,
\begin{equation} \label{eq:map-q}
\hat{q}_{cv}(t) = \argmax p(q_{cv}(t) | \Hcal_v(t)) = \frac{\alpha+n_{cv}(t)-1}{\alpha+\beta+n_{cv}(t)+\bar{n}_{cv}(t)-2},
\end{equation}
\item[II.] via sampling from posterior, \ie,
\begin{equation}
\hat{q}_{cv}(t) \sim p(q_{cv}(t) | \Hcal_v(t)).
\end{equation}
\end{itemize}

Given these estimates, we select what to share next using an empirical
approximation to Eq.~\ref{eq:popt}, \ie,
\begin{equation}
\hat{p}(c | \Hcal(t)) = \II\left[ c = \argmax_{c'} \left( \ab^{\dagger} \hat{\qb} _{c'}(t)+ a_u x_{c'}) \right) \right].
\end{equation}
where $\hat{\qb} _{c'}(t) = (\hat{q}_{c' v}(t))_{v \in \Ncal(u)}$. Algorithm~\ref{alg:sampling} summarizes the complete procedure. Within the algorithm, \textsc{Share} shares a story from a given topic $c_t$, \textsc{GatherFeedback} gathers the feedback from a follower $v$, and \textsc{Estimate} returns an estimate of the preferences parameters using either point estimates or posterior samples.

\begin{algorithm}[t]                    
	\caption{Utility maximization for unknown preferences}
	\label{alg:sampling}
	\begin{algorithmic}[1]
	\STATE \textbf{Input: } Prior parameters $\alpha,\ \beta$.
	\STATE \textbf{Output: } Category $\{c_t\}_{t\in[T]}$.
	\FOR{$v\in\Ncal(u)$}
	\STATE $\Hcal_v(t) \leftarrow \emptyset$;
	\ENDFOR
	\FOR{$t \in \{1, \ldots, T\} $}
	  \STATE $c_t = \argmax_{c'} \big(\sum_{v\in\Ncal(u)} a_v \hat{q}_{c'v}(t)+ a_u x_{c'}\big)$;
	  \STATE $\textsc{Share}(c_t)$;
	  \STATE \texttt{\small /* Gather feedback from $u$'{}s followers */}
	  \FOR{$v\in\Ncal(u)$}
	   \STATE $\Hcal_v(t+1) \leftarrow \Hcal_v(t) \cup \textsc{GatherFeedback}(v)$;
	    \FOR{$c \in \Ccal$}
	      \STATE $\hat{q}_{cv}(t+1) = \textsc{Estimate}(\alpha, \beta, c, \Hcal_v(t+1))$;
	    \ENDFOR
	\ENDFOR
	\ENDFOR
	\end{algorithmic}
\end{algorithm}

\vspace{-1mm}
\subsection{Unknown Preferences:  Analysis}
\vspace{-1mm}
In this section, we analyze the theoretical guarantees of Algorithm~\ref{alg:sampling} in terms of regret $R(T)$, which we define as follows:
\begin{equation}
R(T) = \textsc{Util}(T) - \textsc{Util}^{*}(T),
\end{equation}
where $\textsc{Util}(T)$ is the utility achieved by Algorithm~\ref{alg:sampling} and $\textsc{UTIL}^{*}(T)$ is the utility achieved by the optimal categorical distribution
$p^{*}$, given by Eq.~\ref{eq:popt}, under the true preference parameters $q_{cv}$. Note that an algorithm is called a \emph{no-regret} algorithm if the regret grows sublinearly (\ie, $o(T)$) which implies that the algorithm's average performance converges to that of the optimal algorithm.

Note that the utility $\textsc{Util}(T)$ depends on the quality of the estimates $\hat{q}_{cv}(t)$, which in turn depend on (i) the estimation method (\ie, point estimates vs posterior
samples) and (ii) whether each follower'{}s feedback history $\Hcal_{v}(t)$ only contains the feedback the follower gives to user $u$ or it also contains the feedback she gives
to others. Next, we study several cases separately.


%
\subsubsection{Point estimates}
If the user only has access to the feedback she receives from her followers, we have the following negative result (proven in Appendix~\ref{app:thm:map-wofeedback}):
\begin{theorem}\label{thm:map-wofeedback}
Assume user $u$ uses point estimates $\hat{q}_{cv}(t)$ and she can only access the feedback she receives from her followers. Then, Algorithm~\ref{alg:sampling} suffers
linear  regret $\Theta(T)$.
\end{theorem}
Perhaps surprisingly, Algorithm~\ref{alg:sampling} with point estimates can actually achieve  sublinear regret $O(\sqrt{T})$ if the user has access to both the feedback
her followers give to her as well as to others, as formalized by the following theorem (proven in Appendix~\ref{app:regret-full-feedback}):
%
%
\begin{theorem}\label{thm:map-feedback}
Assume user $u$ uses point estimates for the followers'{} preferences $\hat{q}_{cv}(t)$ and she can access both the feedback her followers give to her as well as to others. Furthermore, the amount of feedback each of her followers $v$ give to others per topic $c$ follows a Poisson distribution with rate $\mu_{cv} > 0$, and let $d = \alpha + \beta$.
Then, for $2 \leq d < 3$, Algorithm~\ref{alg:sampling} achieves  regret
\begin{equation*}
O\left(\underset{c \in \Ccal, v \in \Ncal(u)}{\sum}\frac{\sqrt{T-1}\sqrt{1-e^{-\mu_{cv} (T-1)}}}{2\sqrt{(d-2)\mu_{cv}}} \right),
\end{equation*}
%
%
and, for $d \geq 3$, it achieves:
\begin{equation*}
O\left(\underset{c \in \Ccal, v \in \Ncal(u)}{\sum}\frac{T-1}{2\sqrt{(T-1)\mu_{cv} t+d-3}+\sqrt{d-3}}\right)
\end{equation*}
\end{theorem}
%
Here, note that, whenever her followers do not give feedback to others in at least one topic, \ie, $\sum_{v \in \Ncal(u)} \mu_{cv}=0$ for some $c \in \Ccal$,
Algorithm~\ref{alg:sampling} will suffer linear regret, which is in agreement with Theorem~\ref{thm:map-wofeedback}.

\subsubsection{Posterior samples}
%
If the user uses posterior samples, she is better off. In this case, Algorithm~\ref{alg:sampling} achieves sublinear regret independently on whether $\Hcal_v(t)$ contains
only the feedback she receives or also the feedback her followers give to others. More formally, we have the following Theorem and Corollary (proven in Appendix~\ref{app:regret-full-feedback}):
%
%
\begin{theorem}\label{thm:post}
Assume user $u$ uses posterior samples to estimate the followers'{} preferences $\hat{q}_{cv}(t)$, she can access both the feedback her followers give to her as well as to others,
and the amount of feedback each of her followers $v$ give to others per topic $c$ follows a Poisson distribution with rate $\mu_{cv} > 0$. Then, Algorithm~\ref{alg:sampling}
has regret
%
%
\begin{equation*}
O\left( \log \left(1+\sum_{v\in\Ncal(u)}\frac{1-\exp(-\mu_{cv} \theta_m T)}{\mu_{cv} \theta_m }\right)\right),
\end{equation*}
where $\theta_m$ depends on the parameters $\Qb$ in non-trivial way.
\end{theorem}
\begin{corollary}\label{cor1}
Assume user $u$ uses posterior samples to estimate the followers'{} preferences $\hat{q}_{cv}(t)$ and she can only access the feedback she receives from her
followers. Then, Algorithm~\ref{alg:sampling} has regret $O(\log T)$.
\end{corollary}

In summary, if user $u$ uses posterior samples instead of point estimates, she can effectively trade off exploitation---sharing stories for maximizing her utility---and exploration---sharing 
stories to estimate her followers' preferences. Moreover, if she has access to the feedback that her neighbors provide to others, she is better off using it. Such additional information helps 
her to maximize her utility more effectively in both cases---posterior samples and point estimates.

%% file: 031estimation.tex
%

In this section, assume we observe both the stories user $u$ shared at each time $t \in \{1, \ldots, T\}$ and the feedback she received from her followers.
%
%
Then, our goal is to determine whether the user utilizes the feedback she receives from each of her followers to decide what to post next.
To this aim, we first find the model parameters that best fit the observed data and then determine its statistical significance using statistical 
hypothesis testing.
\vspace*{-1mm}
\subsection{Parameter estimation}
To find the weights $(a_v)_{v\in\Ncal(u)}$ and $a_u$ and parameters $(x_c)_{c \in \Ccal}$ in Eq.~\ref{eq:utility-maximization} that best fit the 
observed data, one could resort to maximum likelihood estimation, \ie,
\begin{equation} \label{eq:utility-estimation}
\begin{aligned}
\underset{\ab, a_u, \xb}{\text{maximize}} & \quad \sum_{t\in[T]} \log p^{*}(c_t | \Hcal(t)) \\
\text{subject to} & \quad a_u \geq 0, a_v \geq 0 \quad \forall v \in \Ncal(u)  \\
& \sum_{v \in \Ncal(u)} a_v + a_u = 1 \\
& \quad 0 \leq x_c \leq 1, \forall c \in \Ccal.
\end{aligned}
\end{equation}
where $p^{*}(c | \Hcal(t))$ implicitly depends on $(a_v)_{v\in\Ncal(u)}$ and $a_u$ since it is the distribution that maximizes the utility function. 
%

%
%
%
%
However, the above maximum likelihood estimation problem faces two serious challenges. 
First, it is stated in terms of either the optimal distribution $p^{*}(c | \Hcal(t))$ or, in practice,
$\hat{p}(c | \Hcal(t))$. However, both distributions concentrate their entire probability mass in one topic.
As a consequence, if there exists a time $t \in \{1, \ldots, T\}$ in which user $u$ shares a story from a topic that
is not the optimal one to choose, the log-likelihood $p^{*}(c | \Hcal(t))$ (or $\log \hat{p}(c | \Hcal(t))$) becomes 
unbounded.
Fortunately, we can overcome this undesirable behavior by approximating the distribution $p^{*}(c | \Hcal(t))$ (or $\hat{p}(c | \Hcal(t))$) 
using a softmax distribution
%
\begin{align}
   p_{\lambda}(c | \Hcal(t)) = \frac{\exp(\lambda(\ab^\dagger  \qb_c(t) + a_u {x}_{c} ))}{\sum_{c' \in \Ccal}\exp(\lambda( \ab^\dagger  \qb_{c'}(t) + a_u {x}_{c'} ))}\label{eq:poptmax}
\end{align}
where $\lambda$ is a given parameter and $\qb_c(t) = \qb_c$ or $\qb_c(t) = \hat{\qb}_c(t)$.
Second, if the number of feedback events is large, the above maximum likelihood estimation problem is not scalable.
To ameliorate this second challenge, in the following, we present a highly efficient heuristic based on linear loss minimization.

Our starting point is the following intuition: at each time $t \in \{1, \ldots, T\}$, if user $u$ is sampling the topic of
the story she shares from a distribution that is close to  ${p}(c | \Hcal(t))$, then the difference between the
optimal topic and the observed topic, \ie,
\begin{equation}\label{eq:approx-obj}
   \max_{c' \in \Ccal} \big(\ab^\dagger  {\qb}_{c'} (t)+a_u x_{j u}\big) - (\ab^\dagger  {\qb}_{c_t} (t)+a_u x_{c_t})
\end{equation}
should be small. Therefore, our heuristic finds the weights $\ab$ and $a_u$ and the parameters $\xb$ that minimizes
this difference over time. 
In practice, the solution to the above problem depends on whether we assume that the user utilized point estimates or posterior samples for
the followers'{} preferences. Therefore, we proceed in turn.

\subsubsection{Point estimates}
If we assume that the user utilized point estimates $\hat{q}_{cv}(t)$ for her followers'{} preferences, then we minimize:
\begin{align}\label{eq:approx-obj-point}
 \sum_{t\in[T]} & \max_{c' \in \Ccal} \big(\ab^\dagger  \hat{\qb}_{c'} (t)+a_u x_{j u}\big) - (\ab^\dagger  \hat{\qb}_{c_t} (t)+a_u x_{c_t}) 
\end{align}
However, the optimization problem is not convex due to the terms $a_u x_{c}$ and, in its current form, it is difficult to solve efficiently. Fortunately, an invertible
nonlinear transformation of the variables transforms it into a convex problem. Let $z_c = a_u x_{c}$ and $\zb=(z_c)_{c\in\Ccal}$. Then, we can rewrite 
the optimization problem as:
%
%
\begin{align*}
\underset{\ab, a_u, \zb}{\text{minimize}} & \sum_{t\in[T]} \left[ \max_{c' \in \Ccal} \left(\ab^\dagger \hat{\qb}_{c'} (t)+z_{c'}\right)- \left(\ab^\dagger \hat{\qb} _{c_t} (t)+z_{c_t} \right) \right],\\
\text{subject to} & \quad a_u \geq 0, a_v \geq 0 \quad \forall v \in \Ncal(u)  \\
& \sum_{v \in \Ncal(u)} a_v + a_u = 1\\
& \quad 0 \leq z_c \leq a_u, \forall c \in \Ccal.
\end{align*}
which is a convex problem jointly in $\ab$, $a_u$, and $\zb$ using composition rules and the fact that the log-exp function is convex.
Once we solve this convex problem, we can recover $x_{c}$ as $x_{c} = z_{c} / a_u$.

\subsubsection{Posterior samples}
If we assume that the user utilized posterior samples $\hat{q}_{cv}(t)$ for her followers'{} preferences, then, we need to take the average of the
objective function with respect to the posterior distributions of the followers'{} preferences for all topics, since we do not know which sample the
user actually took, \ie, the objective function becomes
\begin{equation*}
\sum_{t \in [T]}   \EE_{ \hat{\qb}(t), t \in [T] } \left[ \max_{c' \in \Ccal} \big(\ab^\dagger  \hat{\qb}_{c'} (t)+a_u x_{j u}\big) - (\ab^\dagger  \hat{\qb}_{c_t} (t)+a_u x_{c_t}) \right],
\end{equation*}
where $\hat{q}_{cv}(t) \sim p(q_{cv}(t) | \Hcal_v(t))$ for all $c \in \Ccal$ and $v \in \Ncal(u)$.
%
%
%
In practice, we just replace the above objective function by the empirical average with respect to
the followers'{} preferences, \ie,
\begin{equation*}
\sum_{t\in[T]} \sum_{s=1}^S  \left[ \max_{c' \in \Ccal} \left(\ab^\dagger \hat{\qb}^{(s)}_{c'} (t)+z_{c'}\right)- \left(\ab^\dagger \hat{\qb}^{(s)} _{c_t} (t)+z_{c_t} \right) \right],
\end{equation*}
where $\hat{q}^{(s)}_{cv}(t) \sim p(q_{cv}(t) | \Hcal_v(t))$ for all $c \in \Ccal$ and $v \in \Ncal(u)$ and $S$ is the number of samples.

\vspace{-1mm}
\subsection{Statistical Hypothesis Testing}
\vspace{-1mm}
Given an estimation of the model parameters, we determine their statistical significance using statistical hypothesis testing. More specifically, 
we proceed as follows.
%

Under the null hypothesis $\mathbb{H}_0$, the user does not utilize the feedback she receives from her followers to decide what to post
next, \ie, $a_v = 0$ for all $v \in \Ncal(u)$, and, under the alternative hypothesis, the user does utilize it, \ie, $a_v \geq 0$ for all $v \in \Ncal(u)$.
%
%
Then, for each user $u$, we use the log-likelihood ratio ($LLR$) as a test statistic to measure the statistical power of the feedback data, \ie,
\begin{align}
  LLR(u)=\sum_{t \in [T_u]} \log p_{\lambda, \hat{\ab} \geq 0}(c_t|\Hcal(t))-\sum_{t \in [T_u]} \log p_{\lambda, \hat{\ab} = 0}(c_t|\Hcal(t)),\nn
\end{align}
where $T_u$ is the total number of posts shared by user $u$.
%
%
%
%
Finally, we assess the statistical significance of the $LLR$ values using the theoretical distribution of the $LLR$ under the null hypothesis, \ie, 
$\chi^2 _1(|\Ncal(v)|-1)$, given by Wilks' theorem~\cite{wilks1938large},
where a high value of $LLR$ allows us to reject the null hypothesis with high probability (low $p$-value). 

%% file: 040syn-experiments.tex
%
In this section, we experiment with synthetic data to: (i) illustrate the theoretical properties of Algorithm~\ref{alg:sampling}, discussed
in Section~\ref{sec:properties}; and (ii) show that our utility estimation framework, discussed in Section~\ref{sec:estimation}, can be
used to accurately estimate a user'{}s utility from historical data.

\xhdr{Experimental setup}
We assume that the user $u$ under study has $|\Ncal(u)| = 10$ followers. Then, we sample her weights $a_v \sim \text{Dir}(\gamma)$ and
$a_u \sim Dir(\gamma)$ with $\gamma = 0.8$,
the preferences of her followers $q_{cv}\sim\text{Beta}(0.4,0.6)$,
and her own preferences $x_{c} \sim \text{Beta}(0.4,0.6)$.
The number of topics $|\Ccal| = K$ and rates $\mu_{cv}$ vary for different experiments and thus are specified therein.

\xhdr{Regret analysis}
First, we consider the case when $u$ estimates her followers'{} preferences only on the basis of the feedback her stories receive from
her followers.
To that aim, we fix $\mu_{cv}=0$ for all $c\in\Ccal$ and $v\in\Ncal(u)$, and simulate data from our feedback model of posting behavior
for different number of topics $|\Ccal|=K$. Then, we investigate the variation of the regret over time.
Figure~\ref{fig:reg1} summarizes the results which show that:
(i) point estimates suffer linear regret whereas, posterior samples achieve logarithmic regret, thereby supporting our theoretical findings in
Theorem~\ref{thm:map-wofeedback} and Corollary~\ref{cor1}; and,
(ii) as the number of topics increases, the number of unknown preferences increases, and as a result, the regret increases.

Next, we consider the case when $u$ additionally utilizes the feedback her followers give to others. To that aim, we set $K=10$, sample
$\mu_{cv}\sim \text{Unif}[0,2\bar{\mu}]$ and simulate our model on user $u$ for different value of $\bar{\mu}$.
Figure~\ref{fig:reg2} summarizes the results which show that:
(i) the additional information (\ie, the feedback to others) significantly reduces the regret---even point estimates achieve a regret of $O(\sqrt{T})$,
and moreover, posterior samples achieve a constant regret (\ie, $O(1)$), thereby supporting Theorem~\ref{thm:map-feedback} and Theorem~\ref{thm:post};
and,
(ii) as $\bar{\mu}$ increases, $u$ estimates their followers'{} preferences on the basis of a larger amount of feedback and, as a result, the regret
decreases.
\begin{figure}[t]
\centering
\subfloat[$\hat{q}_{cv}(t)$ are point estimates]{\includegraphics[width=0.26\textwidth]{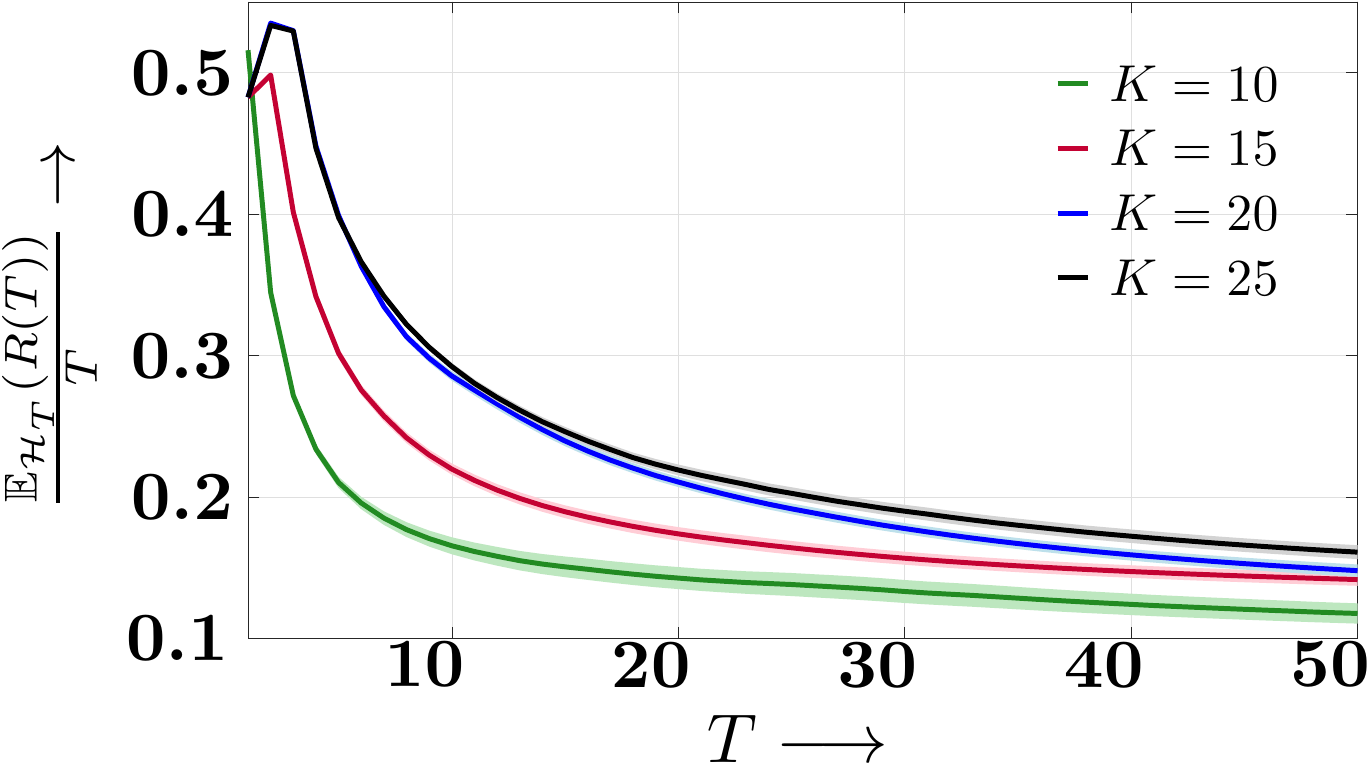}\label{fig:regretMAPWithOutNbrsRatio_new}}\hspace*{1cm}
\subfloat[$\hat{q}_{cv}(t)$ are posterior samples]{\includegraphics[width=0.26\textwidth]{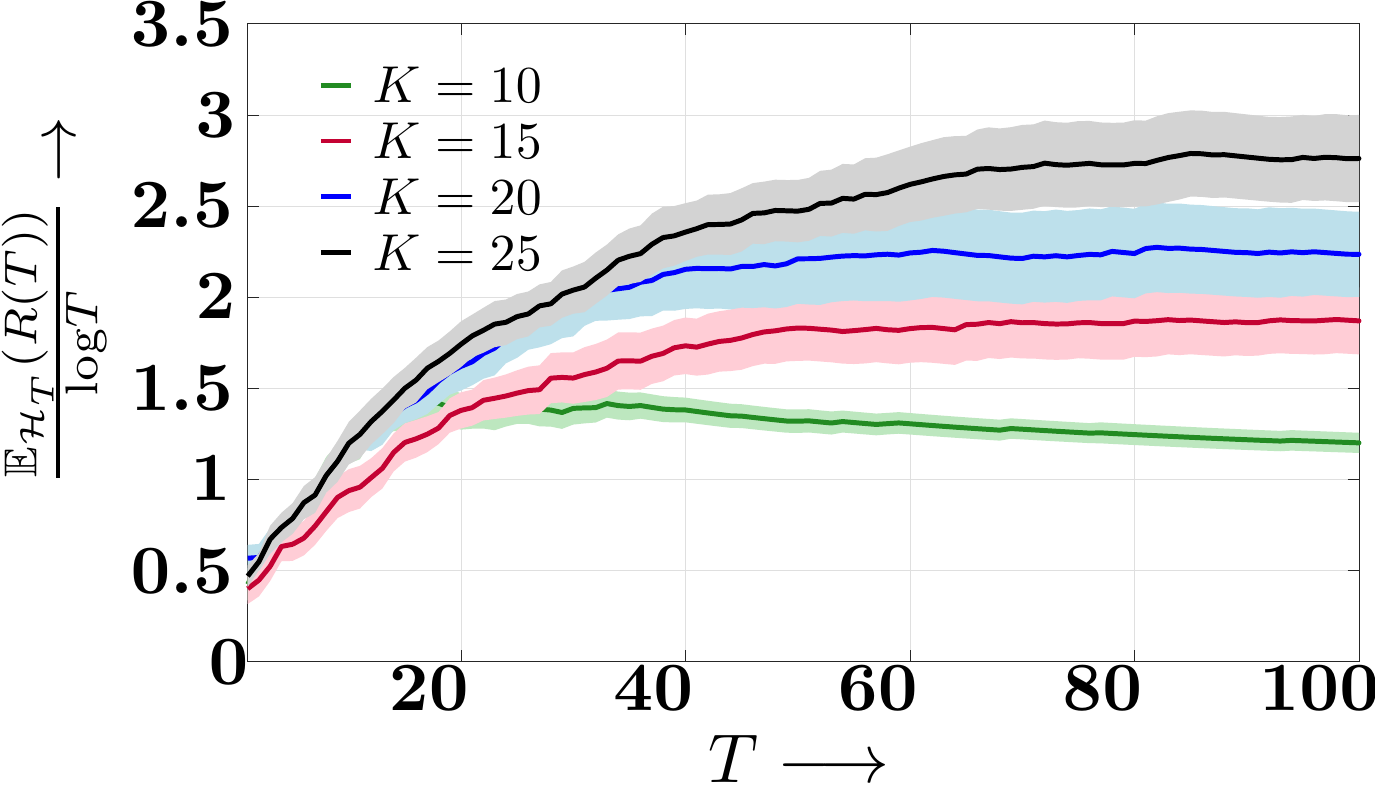}\label{fig:regretPosteriorWithOutNbrs_new}}
\vspace*{-2mm}
\caption{Regret analysis when user $u$ only leverages the feedback she receives, \ie, $\ \mu_{cv}=0$. Panel (a) shows that,
if she uses point estimates for her followers'{} preferences, Algorithm~\ref{alg:sampling} suffers linear regret. Panel (b)
shows that, if she uses posterior sampling, it achieves logarithmic regret.
In both panels, as the number of topics $K$ increases, the regret increases.}
 \label{fig:reg1}
\vspace*{-4mm}
\end{figure}
 \begin{figure}[t]
\centering
\subfloat[$\hat{q}_{cv}(t)$ are point estimates]{\includegraphics[width=0.26\textwidth]{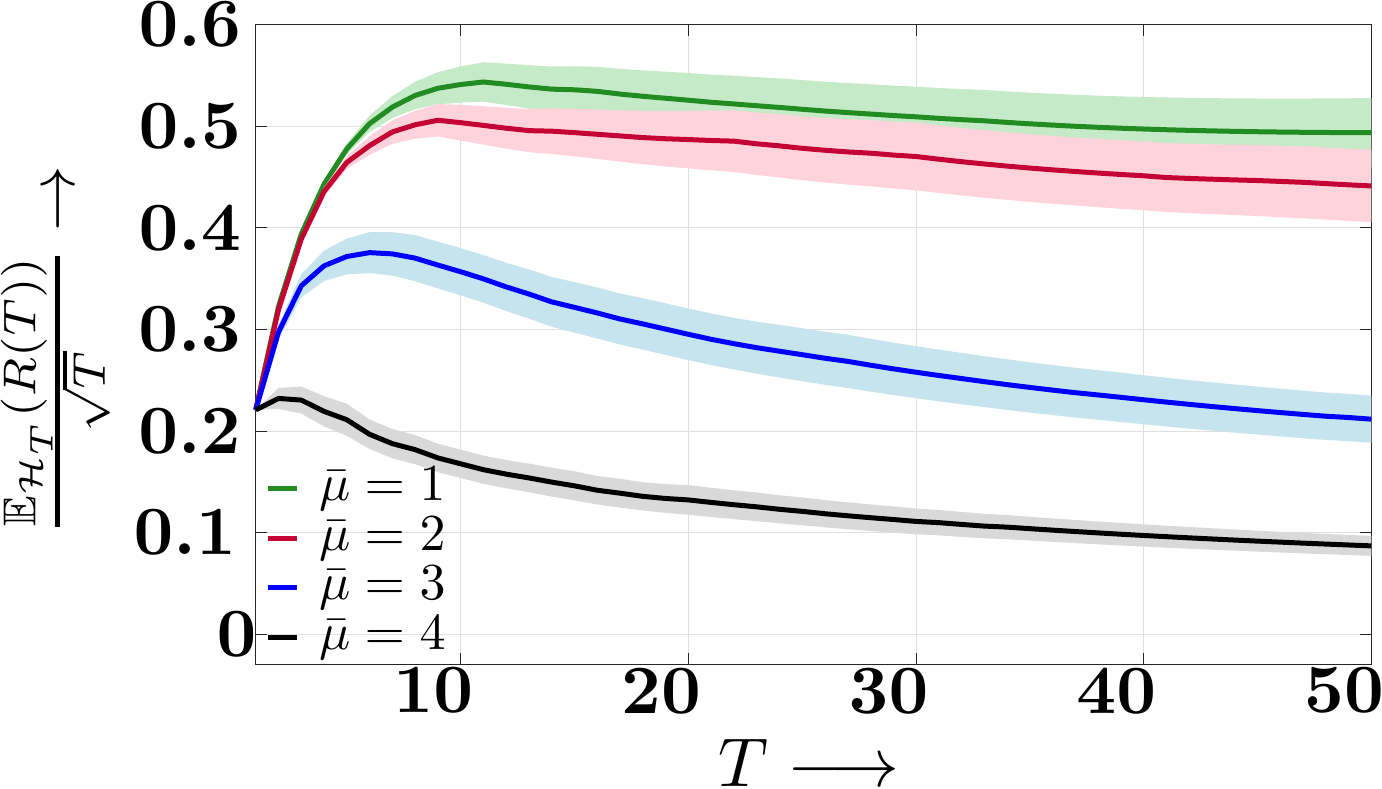}\label{fig:regret1}}\hspace*{1cm}
\subfloat[$\hat{q}_{cv}(t)$ are posterior samples]{\includegraphics[width=0.26\textwidth]{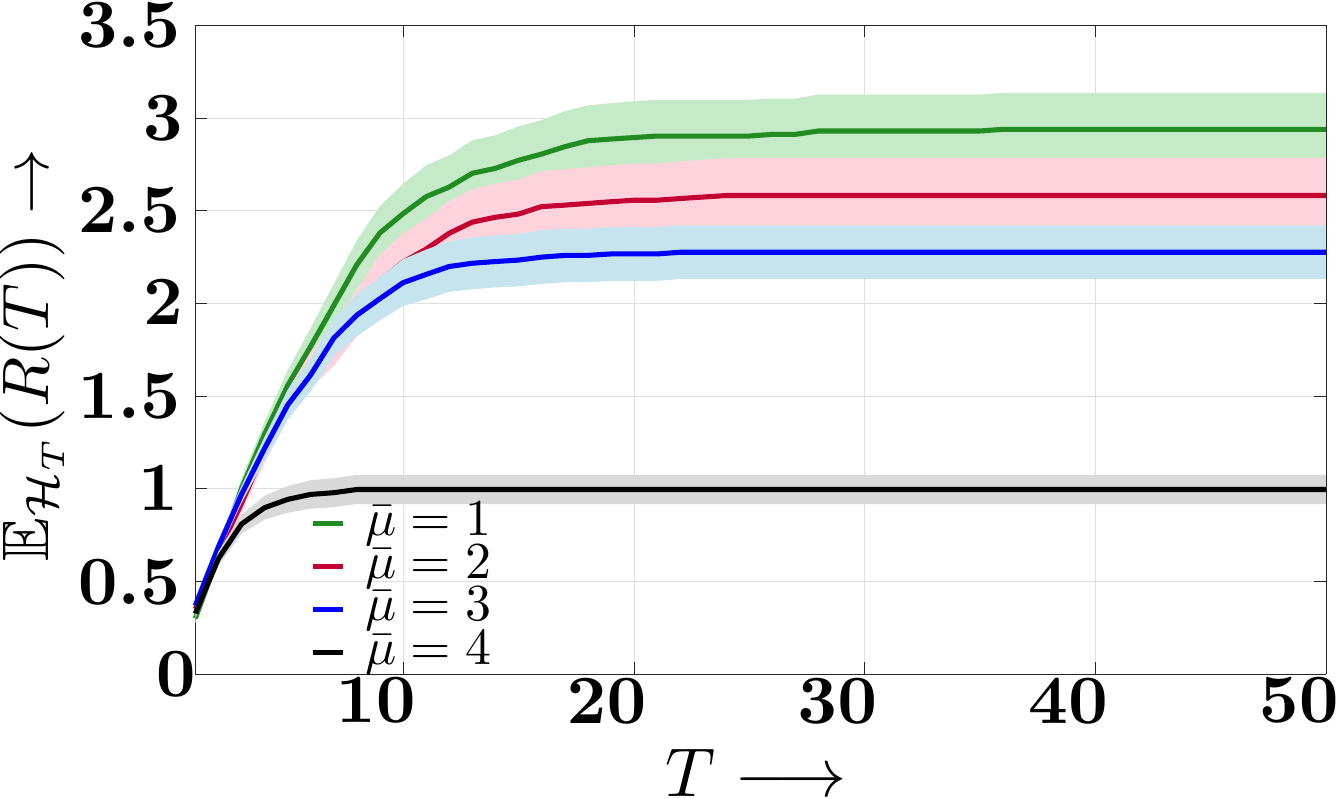}\label{fig:regretPosteriorWithOutNbrs_new}}
\vspace*{-2mm}
\caption{Regret analysis when user $u$ levegares the feedback she receives as well as the feedback her followers give to others, \ie, $\mu_{cv}>0$.
Panel (a) shows that, if she point estimates for her followers'{} preferences, Algorithm~\ref{alg:sampling} achieves sublinear regret $O(\sqrt{T})$.
Panel (b) shows that, if she uses posterior sampling, it achieves constant regret $O(1)$.
In both panels, as the average rate of feedback to others $\mu_{cv}$ increases, the regret decreases.}
 \label{fig:reg2}
\vspace*{-4mm}
\end{figure}

\xhdr{Model estimation}
To investigate the performance of our utility estimation framework, we first generate $\Hcal(T)$ by simulating data from our model with
$K=10$ and sample $\mu_{cv}\sim\text{Unif}[0,2\bar{\mu}]$.
Then, we train our model using the generated $\Hcal(T)$, for different $T$ and $\bar{\mu}$ values using our two estimation methods from
Section~\ref{sec:estimation}.
Finally, we evaluate the accuracy of model estimation procedures in terms of root mean square error (RMSE) between the estimated and true parameters,
\ie, $$RMSE = \sqrt{\EE(||\ab-\hat{\ab}||^2) + \EE(||a_u-\hat{a}_u||^2) + \EE(||\xb_u-\widehat{\xb}_u||^2)}.$$ 
Figure~\ref{fig:est} summarizes the results the method based on linear loss minimization, which show that,
(i) as $T$ increases and we feed more training samples into the estimation procedure, the accuracy increases;
(ii) similarly, as $\bar{\mu}$ increases and we feed more feedback into the estimation procedure, the accuracy increases; and,
(iii) the estimation accuracy for posterior samples is significantly better than for point estimates;

%% file: 041real-experiments.tex
In this section, we apply our utility estimation algorithms to several real datasets gathered from Twitter and Reddit and then, using the utility estimation 
framework described in Section~\ref{sec:estimation}, show that $53\%-82\%$ of the users in the Twitter datasets and $28\%-43\%$ of the users in the 
Reddit datasets use the feedback they receive from their followers to decide what to post next.

\begin{figure}[t]
\centering
\subfloat[$\hat{q}_{cv}$ are point estimates]{\includegraphics[width=0.26\textwidth]{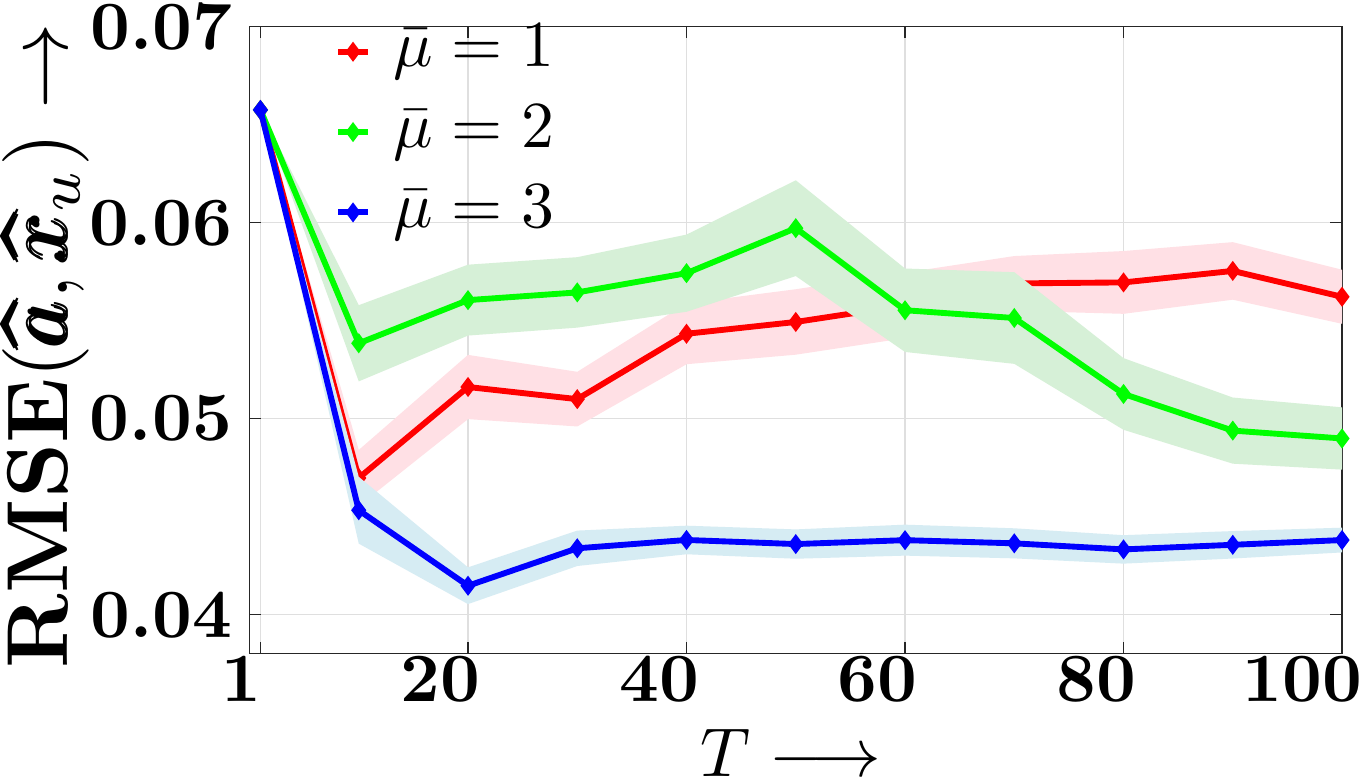}\label{fig:regretPosteriorWithOutNbrs_new}}\hspace*{1cm}
\subfloat[$\hat{q}_{cv}$ are posterior samples]{\includegraphics[width=0.26\textwidth]{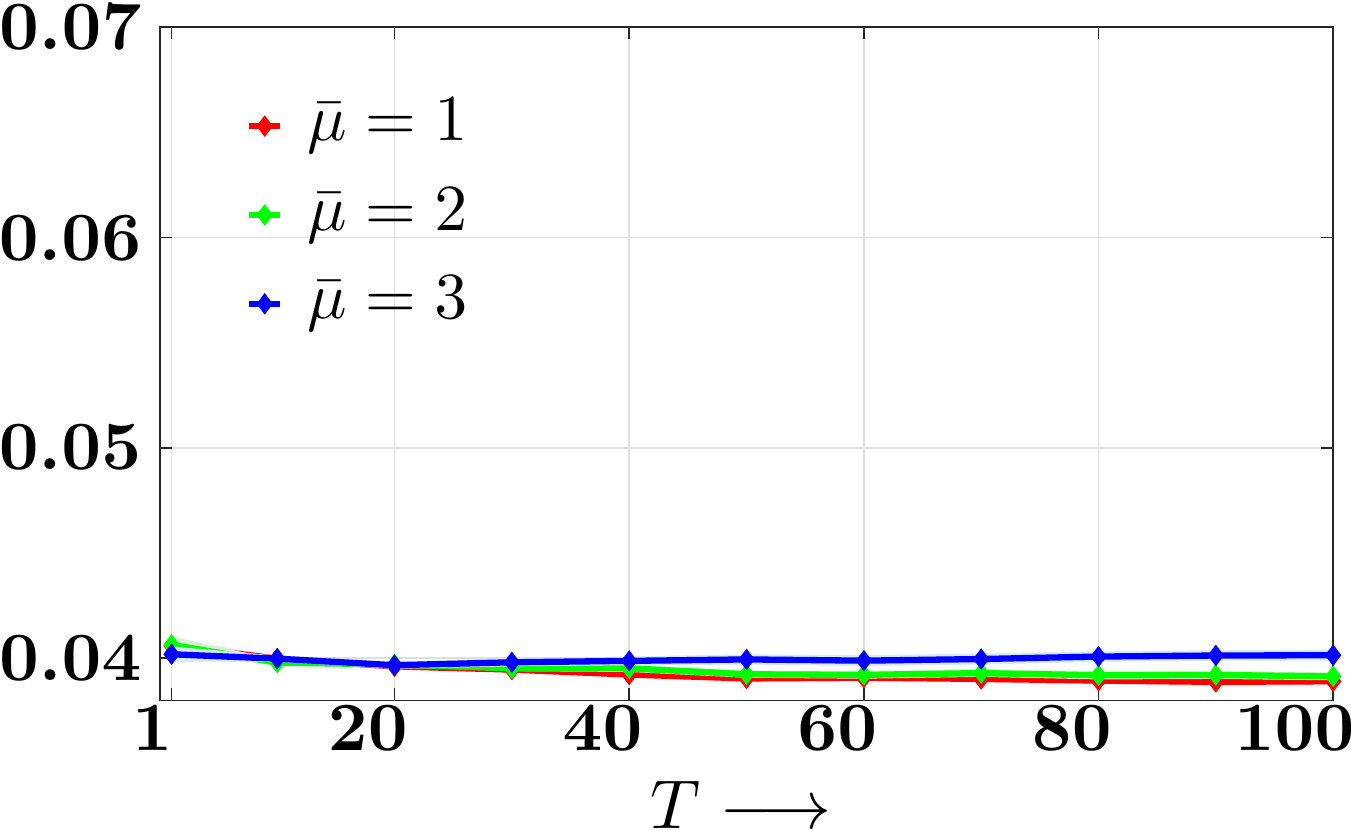}\label{fig:regretPosteriorWithOutNbrs_new}}
\vspace*{-2mm}
 \caption{Performance of our utility estimation method based on linear loss minimization in terms of $RMSE = \sqrt{\EE(||\ab-\hat{\ab}||^2) + \EE(||a_u-\hat{a}_u||^2) + \EE(||\xb_u-\widehat{\xb}_u||^2)}$ 
 as $T$ increase. The performance is significantly better whenever the user leverages posterior samples for her followers'{} preferences.} 
 \label{fig:est}
\vspace*{-4mm}
\end{figure}

 \begin{figure*}[!t]
\centering
\subfloat[Brazil]{\includegraphics[width=0.23\textwidth]{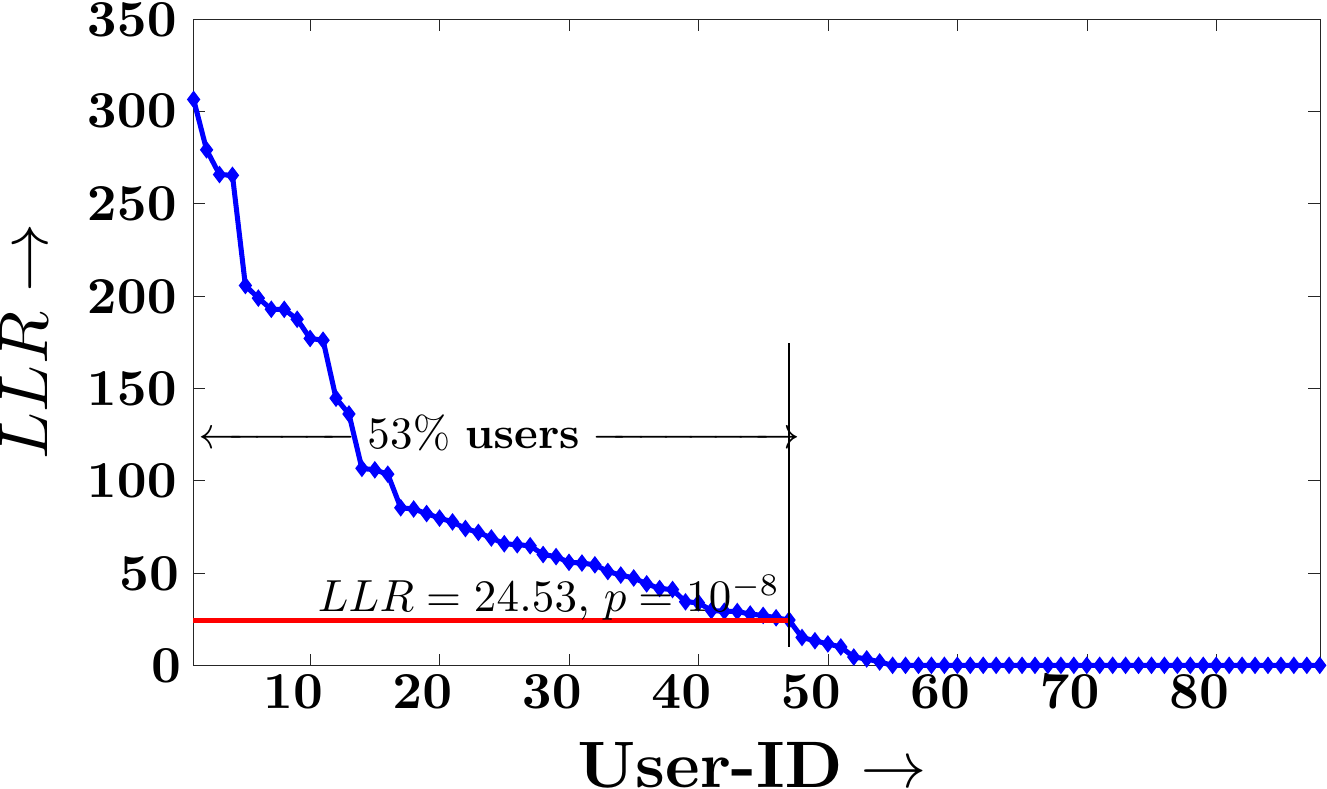}\label{fig:brazil_hyp}}\hspace*{1cm}
\subfloat[TOT]{\includegraphics[width=0.23\textwidth]{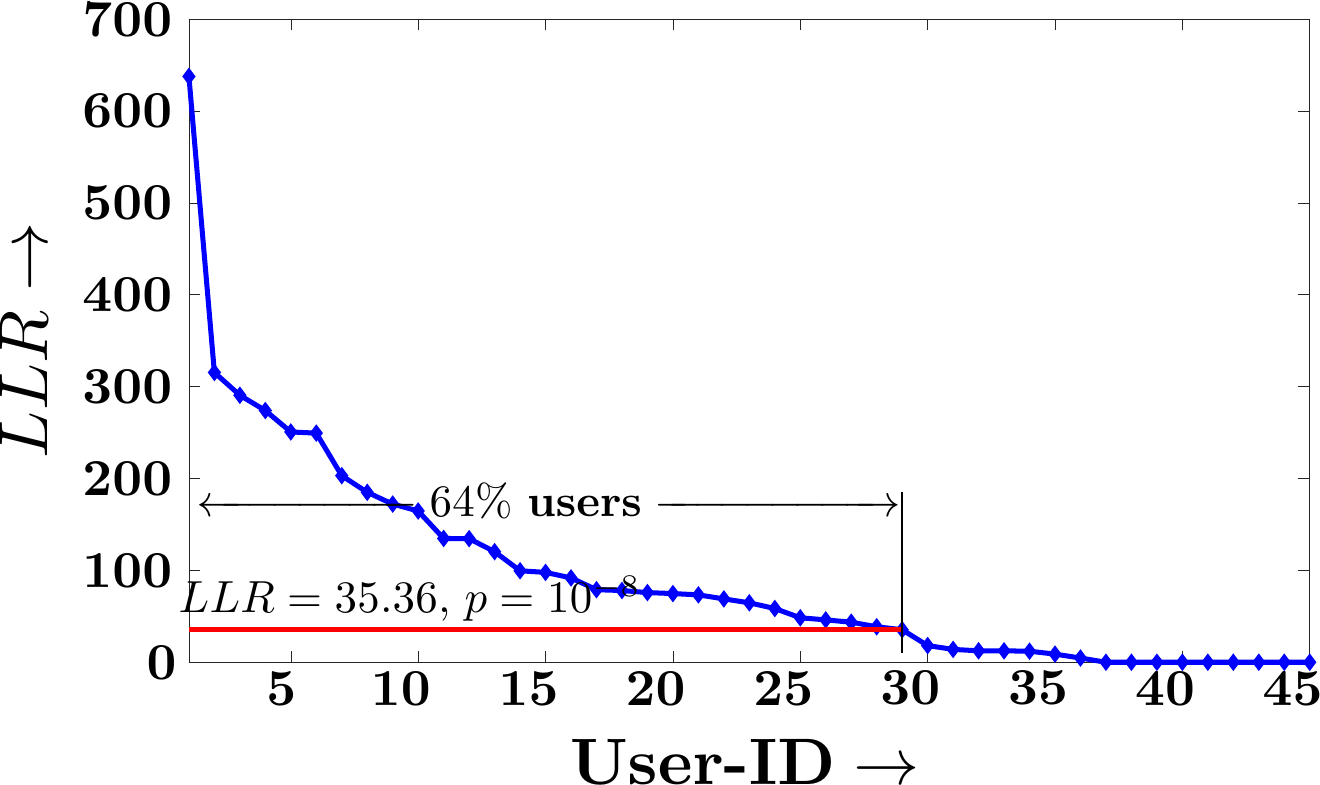}\label{fig:TOT_hyp}}\hspace*{1cm}
\subfloat[Iran]{\includegraphics[width=0.23\textwidth]{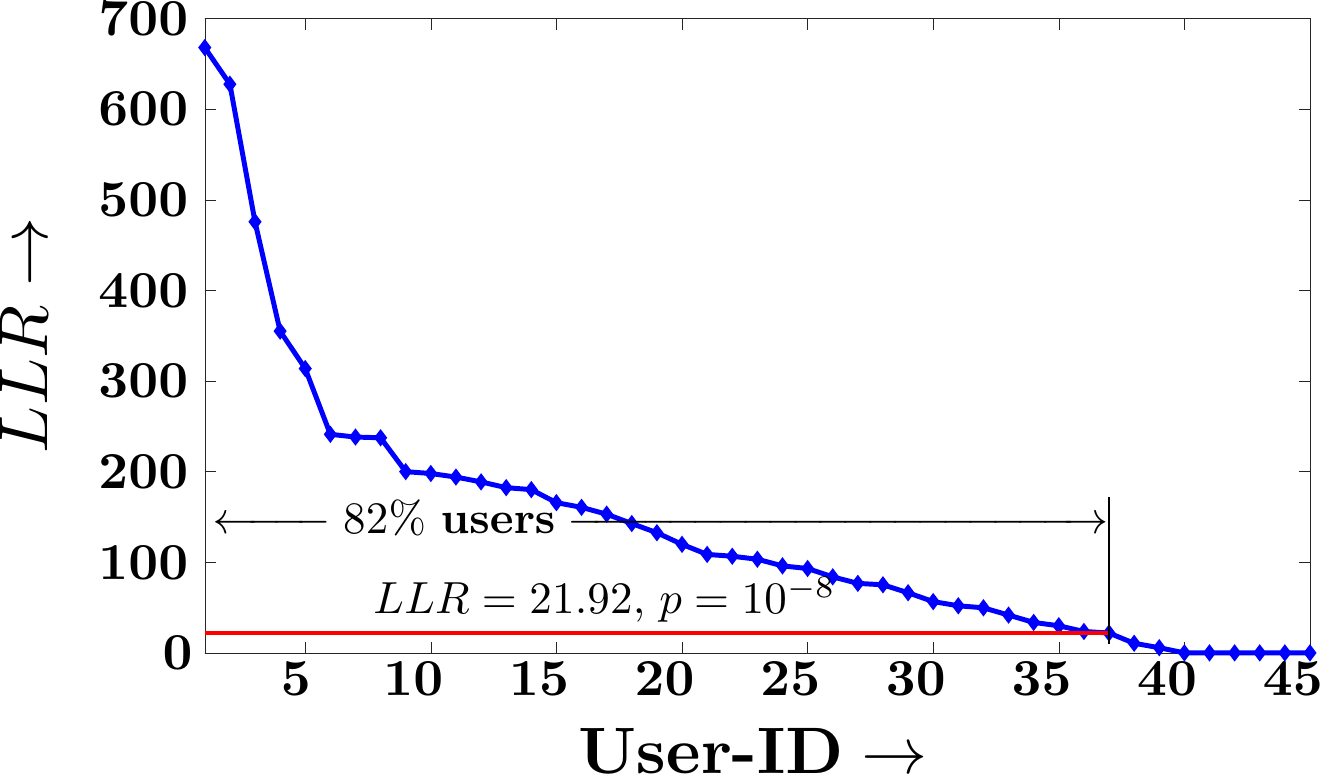}\label{fig:iran_hyp}}\\
\vspace*{-0.3cm}
 \hspace*{1cm}
\subfloat[Leisure]{\includegraphics[width=0.23\textwidth]{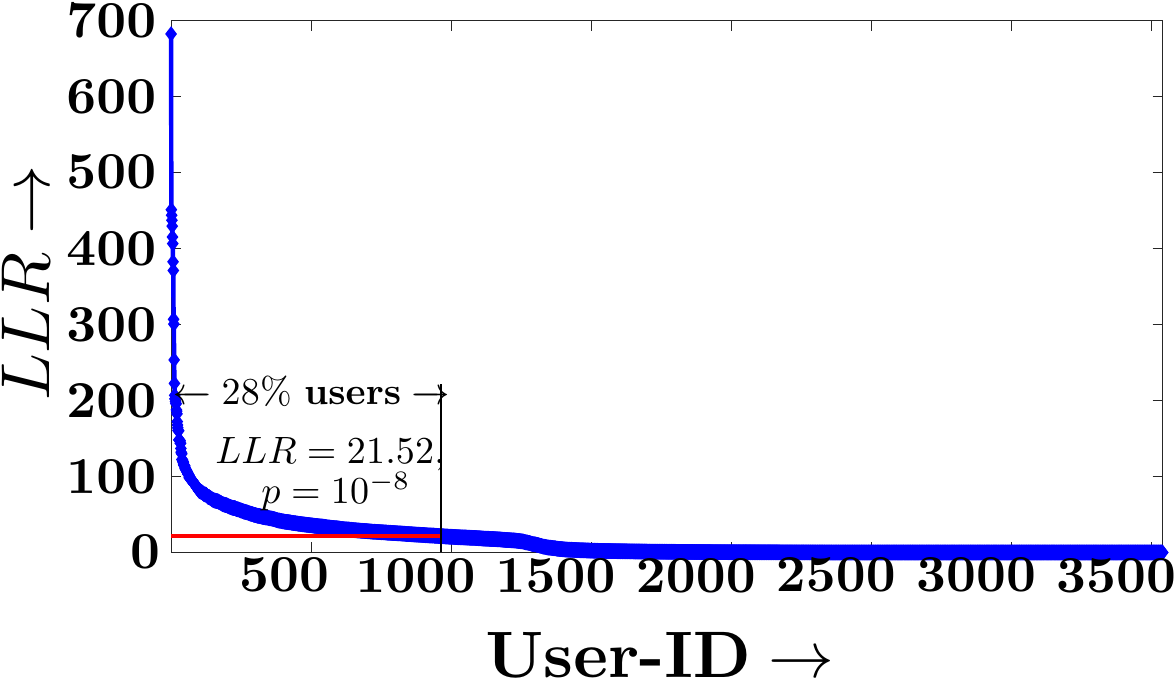}\label{fig:leisure_hyp}}\hspace*{1cm}
 \subfloat[Sports]{\includegraphics[width=0.23\textwidth]{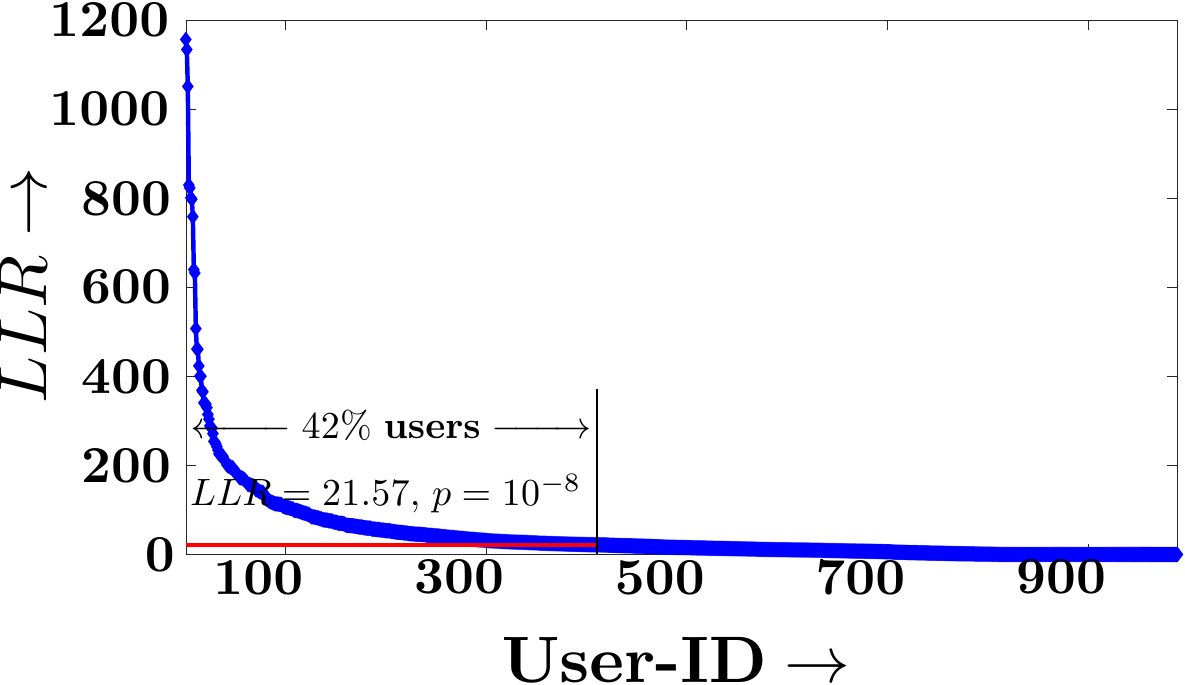}\label{fig:sports_hyp}}\hspace*{1cm}
\subfloat[Learning]{\includegraphics[width=0.23\textwidth]{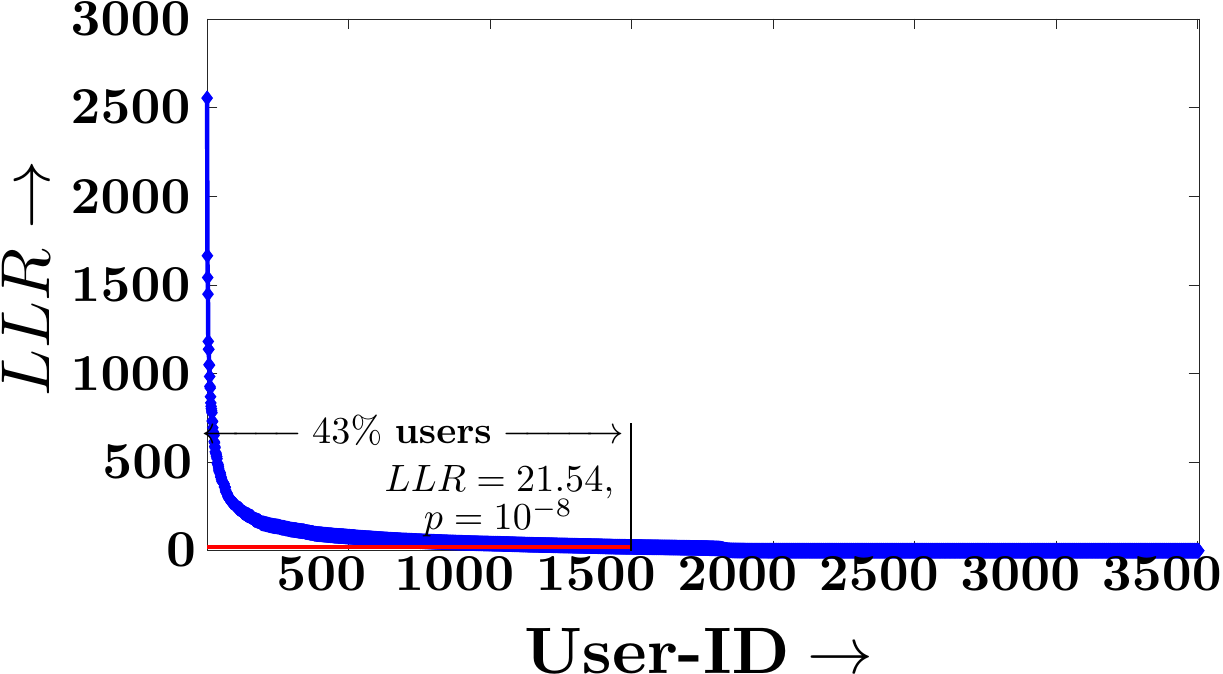}\label{fig:learn_hyp}}\hspace*{1cm}
  \vspace{-3.5mm}
\caption{Log-likelihood ratio ($LLR$) values for all users in the three Twitter (first row) and Reddit (second) datasets. The horizontal red line shows the minimum LLR value to achieve statistical
significance at $p = 10^{-8}$. The results indicate that we can reject the hypothesis that
users do not utilize the feedback they receive from their followers for $53$\%, $64$\% and $82$\%  ($28$\%, $42$\% and $43$\%) of the users, respectively for Twitter and Reddit datasets. }
 \label{fig:emp}
\vspace*{-3.5mm}
\end{figure*}

\xhdr{Data description and experimental setup}
We collect Twitter and Reddit data for evaluating our utility estimation methods.

\emph{--- Twitter: } We used data gathered from Twitter as reported in previous work~\cite{cha2010measuring}, which comprises the profiles of $52$ million users, $1.9$ billion directed
follow links among them, and $1.7$ billion public tweets posted by these users, where the underlying link information is based on a snapshot taken at the time of
data collection, in September 2009.
Here, we focused on the tweets published during a two month period, from July 1, 2009 to September 26, 2009, which allows us to consider the set of followers of a user
to be approximately static.

In our experiments, a follower $v$ provides feedback to a tweet published by a user $u \in \Vcal$ if she retweets it\footnote{\scriptsize Since back in 2009, Twitter did not
have a retweet button, we consider a Jaccard similarity $>80\%$ between tokens contained in two tweets to decide if the latter is a retweet of former.}
and each topic $c \in \Ccal$ corresponds to the most common\footnote{\scriptsize In case a tweet contains more than one hashtag, we consider the hashtag that is more
common across our dataset.} hashtag a tweet contains.
Moreover, using manual inspection, we tracked down the hashtags used in three different themes to create three datasets:
\begin{itemize}[leftmargin=0.7cm]
 \item[---] \textbf{Brazil:} Brazil elections which took place in latter 2010, where $|\Ccal|=4$, $|\Vcal|=88$ and $\EE[T] = 166$.
 \item[---] \textbf{TOT:}  Top US conservatives and liberals on Twitter, where $|\Ccal|=11$, $|\Vcal|=45$ and $\EE[T] = 189$.
  \item[---] \textbf{Iran:} Iran presidential elections in 2009, where $|\Ccal|=5$, $|\Vcal|=45$ and $\EE[T] = 395$.
\end{itemize}
In each of the above datasets, we filtered out hashtags $c$ that were used less than $500$ times and users $u$ who posted less than four tweets with at least two of these
hashtags or whose tweets were not retweeted more than four times by at least 2 followers.
Moreover, for each of user $u \in \Vcal$, we tracked down the five followers who retweeted her tweets more frequently and, for each these followers, we reconstructed the
feedback they provided to her and others by collecting all their retweets as well as the tweets posted by all the users they follow.

\noindent \emph{--- Reddit: } 
We used publicly available data gathered from Reddit\footnote{\scriptsize https://archive.org/details/2015\_reddit\_comments\_corpus}, which comprises the profiles of $5$ million users
and $226$ million comments posted by these users in the month of May, 2015.
%
In our experiments, a user $v$ provides feedback to a message published by a user $u$ if she replied to it and each topic $c\in \Ccal$ corresponds to the 
subreddit in which a comment was written.
Here, we tracked down the subreddits in three different themes to create three datasets:
%
%
%
\begin{itemize}[leftmargin=0.7cm]
 \item[---] \textbf{Leisure:} \textit{r/funny}, \textit{r/pics} and \textit{r/WTF}, where $|\Ccal|=3$, $|\Vcal|=3,540$ and $\EE[T] = 126$.
 \item[---] \textbf{Sports:} \textit{r/CFB}, \textit{r/nba} and \textit{r/nfl}, where $|\Ccal|=3$, $|\Vcal|=989$ and $\EE[T] = 309$.
 \item[---] \textbf{Learning:} \textit{r/AdviceAnimals}, \textit{r/TodayILearned}, \textit{r/AskReddit} and \\ \textit{r/worldnews}, where $|\Ccal|=4$, $|\Vcal|=3,505$ and $\EE[T] = 218$.
\end{itemize}
In each of the above datasets, we only considered users $\Vcal$ who have made at least 20 top-level comments in at least 2 distinct categories $c \in \Ccal$.
Moreover, for these users $u \in \Vcal$, we tracked down the five users who have replied to their comments more frequently and consider them to be the neighbors of $u$.
Finally, for each of these followers, we reconstructed the feedback they provided to $u$ and others by collecting all their replies as well as the comments posted by all their neighbors.

\xhdr{Results}
We determined whether each user in each of the six datasets utilizes the feedback she receives from each of her follo\-wers to decide what to post next using the 
utility estimation framework described in Section~\ref{sec:estimation}.
Figure~\ref{fig:emp} summarizes the results, which show that, at p-value $p = 10^{-8}$, we can reject the hypothesis that users do not utilize the feedback they receive
from their followers for $53$\%, $64$\%, $82$\% of the users in the Twitter datasets and for $28$\%, $42$\%, $43$\% of the users in the Reddit datasets. 

%% file: 050conclusion.tex
%
%
In this paper, we introduced a feedback model of posting behavior in social media. The model allowed us: (i) to investigate under which
conditions can a user succeed at maximizing the (positive) feedback she receives; and, (ii) to determine whether a user utilizes the feedback
she receives from each of her followers to decide what to post next using observational data. 
%
%
%
%
Moreover, we performed experiments on synthetic and real data gathered from Twitter and Reddit to illustrate our theoretical findings, show
that our estimation methods are able to accurately recover users'{} underlying utility functions, and provide empirical evidence
that
$53\%-82\%$ of the users in the Twitter datasets and $28\%-43\%$ of the users in the Reddit datasets use the feedback they receive from their followers to decide what to
post next. 

There are many interesting venues for future work. For example, we have assumed that the followers'{} preferences are not influenced
by the users'{} posting behavior. It would be interesting to analyze a scenario in which both users and the followers influence each other.
We have considered a simple linear utility function, a natural next step would be considering more complex utility functions with
higher predictive power.
Finally, it would be very interesting to apply our utility estimation methods to real data from other social media platforms, \eg,
Facebook.

%% file: appen.tex
%
%
\section{Appendix}

\subsection{Proof of Theorem~\ref{thm:map-wofeedback}} \label{app:thm:map-wofeedback}
Consider a special case, where the number of topics is $K=2$; $u$ has only one follower $v$ with $q_{1v} > q_{2v} = \frac{1}{2}$;
the weights $a_v=a_u={1/2}$; the user preferences $x_{1u} = x_{2u} = q \ge 0$; and the prior parameters
$\alpha=\beta=3$.
Then, the regret $R(T)$ is
\begin{align}
\frac{1}{2}\sum_{t\in[T]}\big(&\max  \{q_{1v}+q,q_{2v}+q\} -(q_{c_t v}+q)\big) \nn
\end{align}
Since, $q_{1v} > q_{2v}$, we have $R(T)=T_2(q_{1v}-q_{2v})/2$, where $T_2$ is the number of stories about topic $2$.
Now, $\alpha=\beta$ makes the initial estimates $\hat{q} _{1v}(1)= \hat{q} _{2v} (1)=\frac{1}{2}$, and therefore, a topic is chosen randomly. Let us assume that
topic $2$ (the wrong category with minimum utility) is selected and a story from that topic is shared.
If $v$ likes this story, which can happen with probability $q_{2v}=1/2$, then $\hat{q}_{2v}(2)=(1+3-1)/(1+6-2)=3/5>\hat{q}_{1v}(2)=1/2$ and topic $2$ is again
selected at next time $t=2$.
At $t=3$, $\hat{q}_{2v}(t)$ again increases (decreases) with probability $1/2$ $(1/2)$.
Note that, $u$ keeps choosing topic $2$ as long as $\hat{q}_{2v}(t)> 1/2$ \ie\ $n_{2v}(t)>\bar{n}_{2v}(t)$, and the possibility of selecting topic $1$ only arises
when $n_{2v}(t)=\bar{n}_{2v}(t)$.
Such a situation can be mapped to an instance of \emph{a simple one dimensional random walk}, where the walker starts from origin at $t=1$, if a story from topic
$2$ is shared.
The walker moves to right (left) if user $v$ does (not) like the story. Now, the expected time of the first return to origin in a simple random walk is infinite. Therefore,
once $u$ starts posting the messages with category $2$, the expected time that $n_{1v}(t)=n_{2v}(t)$ for the first time $t=t_{\text{first}}\to \infty$~\cite{feller1968introduction}. Therefore,
$\EE(T_2)=\Theta(T)$. More formally, we have:
\begin{align}
\EE(&T_2)> \ \EE(T_2|c_1=2,l_v(1)=1)\PP(c_1=2,l_v(1)=1) \nn\\
     &   =\EE(T_2|c_1=2,l_v(1)=1)\PP(l_v(1)=1|c_1=2)\PP(c_1=2)\nn\\
     &   = \frac{1}{4}\EE(T_2|c_1=2,l_v(1)=1)
\end{align}
Now, in the random walk setting, $\EE(T_2|c_1=2,l_v(1)=1)$ is the amount of time the walker stays on the positive or right side of the line, and therefore,
greater than the time of first return to origin, which is infinite.
Therefore, the walker stays on the right side for the entire $T$, given the first step is taken towards the right side is $T$.
Hence, we have $\EE(T_2|c_1=2,l_v(1)=1)= T$. hence $\EE(T_2)>T/4$.
%
\vspace*{-0.1cm}
\subsection{Proofs of Theorems~\ref{thm:map-feedback} and~\ref{thm:post}} \label{app:regret-full-feedback}
\vspace*{-0.1cm}
\xhdr{Preliminaries} To prove these theorems, we first present few notations which will be used throughout both the proofs.
For the sake of brevity, we define (i) $q_{cu}\coloneqq x_{c}$; 
%
(ii) $\hat{q}_{cu}(t)=x_c$ for all $t$;
(iii) $\Ncal^*(u)=\Ncal(u)\cup\{u\}$; (iv) $k_c(t)$ as the number of posts made by broadcaster $u$ up to and including time $t$, that have category $c$; 
(iv) $M_{cv}(t)\coloneqq N_{cv}(t)-k_c(t-1)$;
(v) $M_{cv}(t)$ as the number of messages with category $c$, which $u$ may observe as they are appearing the wall of $v$ until and excluding time $t$;
(vi) $F^{Beta}_{\alpha,\beta}(.)$ as the c.d.f of Beta distribution with parameters $\alpha$ and $\beta$; 
(vii) $F^B _{n,p}(.)$ as the c.d.f in Binomial distribution with parameters $n,p$;
(viii) $f_{n,p}(.)$ as the probability mass function in Binomial distribution function; 
(ix) $\hat{\eta} _{cv}(t)\coloneqq n_{cv}(t)/N_{cv}(t)$;
(x) $\hat{\eta}_i (t)\coloneqq\sum_{v\in\Ncal(u)}a_v \hat{\eta} _{cv}(t)+a_u x_c$;
(xi) $\eta_c \coloneqq \sum_{v\in\Ncal(u)}a_v q _{cv}+a_u x_c$;
(xii) $\theta_c\coloneqq \sum_{v\in\Ncal(u)}a_v \qh  (t)+a_u x_c$;  
(xiii) $p_{c,t}\coloneqq\PP(\theta_1(t)>\rho_c+a_u x_c|\Hcal_{t})$; and
(xiv) $d_{KL}(a,b)=a\log (a/b)+(1-a)\log ((1-a)/(1-b))$, where $a,b\in [0,1]$.\\
Furthermore, we define 
\vspace*{-0.1cm}
\begin{align}
& \dT\coloneqq \sum_{c\in\Ccal} \Big(\hspace*{-0.2cm}\sum_{v\in\Ncal(u)} a_v q_{cv}+a_u x_c\Big) \big(p^{*}(c_t|\Hcal_t) -\hat{p} (c_t|\Hcal_t)\big)\nn\\
&\quad\quad\quad\overset{a}{=} \sum_{c\in\Ccal} \sum_{v\in\Ncal^*(u)} a_v q_{cv} \big(p^{*}(c_t|\Hcal_t) -\hat{p} (c_t|\Hcal_t)\big)\label{eq:delHdef}
\end{align}
Equality $a$ is due to definitions (ii) $q_{cu}=x_c$ and  (iii) $\Ncal^*(u)=\Ncal(u)\cup\{u\}$.
\vspace*{-0.1cm}
\begin{align}
 &\text{So, }  \EE_{\Hcal_T} (\dT)= \EE_{\Hcal_t} \EE_{\Hcal_T\backslash \Hcal_t} (\dT)
 =\EE_{\Hcal_t} (\dT)\nn\\
&\text{Hence, }
 \EE_{\Hcal_T}(\regT)= \sum_{t\in[T]}\EE_{\Hcal_t} (\dT)
\end{align}
\xhdr{Proof of Theorem~\ref{thm:map-feedback}}
We have,
\vspace*{-0.2cm}
\begin{align}
\dT=&\max_{c\in\Ccal}\sum_{v\in\Ncal^*(u)} a_v q_{cv} -\sum_{v\in\Ncal^*(u)} a_v q_{c_t v}\nn\\
=&\max_{c\in\Ccal}\sum_{v\in\Ncal^*(u)} a_v q_{cv}-\max_{c'\in\Ccal}\sum_{v\in\Ncal^*(u)} a_v \hat{q}_{c'v}(t)\nn \\
+&\max_{c'\in\Ccal}\sum_{v\in\Ncal^*(u)} a_v \hat{q}_{c'v}(t)-\sum_{v\in\Ncal^*(u)} a_v q_{c_t v}\nn \\
\overset{a}{\le} & \max_{c\in\Ccal} \Big| \hspace*{-0.1cm}\sum_{v\in\Ncal(u)}\hspace*{-0.2cm}\big(a_v q_{cv}-a_v\hat{q}_{cv}(t)\big)\Big|
+\hspace*{-0.2cm}\sum_{v\in\Ncal(u)} \hspace*{-0.2cm}\big(a_v \hat{q} _{c_t v} (t)-a_v q_{c_t v}\big)\nn\\
 {\le}\ &  \ 2\sum_{c\in\Ccal}\sum_{v\in\Ncal(u)}|q_{cv}-\qh(t)|\label{eq:edt1}
\end{align}
The inequality $a$ is obtained using (i) the triangle inequality of max norm; (ii) the fact that $c_t=\argmax_{c'\in\Ccal} \sum_{v} a_v \hat{q}_{c'v}(t)$; and
(iii) $\hat{q}_{cu} (t)=q _{cu}$, where (iii) reduces the summation over $\Ncal^*(u)$ to $\Ncal(u)$. Since, $\EE(X)\le \sqrt{\EE(|X|^2)}$, we have
\begin{align}
\hspace*{-0.4cm}\EE_{\Hcal_t}( \dT) &\le \ 2\sum_{c\in\Ccal}\sum_{v\in\Ncal(u)}\sqrt{ \EE_{\Hcal_t}(q_{cv}-\qh(t))^2}\nn\\
&=2 \hspace*{-0.4cm}\sum_{c\in\Ccal, v\in\Ncal(u)}\hspace*{-0.2cm}\sqrt{ \EE_{N_{cv}(t)}\EE ((\qcv-\qh(t))^2|N_{cv}(t))}\label{eq:medx}
\end{align}
$(\qcv-\qh(t))^2$ can be written as,
\begin{align}
&\qcv^2-2\qcv \frac{\alpha+\ncv-1}{\alpha+\beta+\nv-2}+\Big(\frac{\alpha+\ncv-1}{\alpha+\beta+\nv-2}\Big)^2\nn
\end{align}
On expanding, we have $ \EE((\qcv-\qh(t))^2|N_{cv}(t))$ to be same as  
\begin{align}
\frac{((\alpha-1)-\qcv(\at-2))^2+\nv\qcv(1-\qcv)}{(\alpha+\beta+\nv-2)^2}\nn\\
=\frac{\kappa_1}{(\at+\nv-2)^2}+\frac{\kappa_2}{\at+\nv-2}
\end{align}
where $\kappa_1=((\alpha-1)-\qcv(\at-2))^2-\qcv(1-\qcv)(\at-2)$ and $\kappa_2=\qcv(1-\qcv)$. Finally we observe that,
$(\at+\nv-2)^2 \ge (\at-2)(\at+\nv-2)$. Hence,
\begin{align}
 \EE((\qcv-\qh(t))^2|N_{cv}(t))&\le \frac{\kappa_1/(\at-2)+\kappa_2}{\at+\nv-2}
\end{align}
From Eq.~\ref{eq:medx}, we now have,
\begin{align}
\EE_{\Hcal_t}( \dT) &\le 2\sum_{c\in\Ccal}\sum_{v\in\Ncal(u)}\hspace*{-0.2cm}\sqrt{ \EE_{N_{cv }(t)} \frac{\kappa_1/(\at-2)+\kappa_2}{\at+\nv-2}  }\nn\\
&\le 2\sum_{c\in\Ccal}\sum_{v\in\Ncal(u)}\hspace*{-0.2cm}\sqrt{ \EE_{M_{cv }(t)} \frac{\kappa_1/(\at-2)+\kappa_2}{\at+\mv-2}  }
\end{align}
%
Then we apply Lemma~\ref{lem:key2} to obtain the required bound.

\xhdr{Proof of Theorem~\ref{thm:post}}
To prove this theorem, we leverage the proof techniques of Agarwal \emph{et al} in~\cite{tspaper}. 
%
Without loss of generality, we assume that the set of categories $\Ccal=[K]$; $c=1$ is the optimum true category with maximum utility. We also denote
$\tau_m$ denotes the time step at which a message with optimal category (\ie, category $1$) is posted for the $m$-th time.
In this context, we denote,
\begin{align}
L_c(T):=\frac{1}{d_m}\log\big(1+\sum_{v\in\Ncal(u)}\frac{1-\exp(-\mu_{cv} (1-e^{-d_m}) T)}{\mu_{cv} (1-e^{-d_m})}\big)
\end{align}
%
%
%
%
Then, we define two numbers $\sigma_c,\rho_c$ so that,
\begin{align}
 \eta_c<\sigma_c+a_u x_{c}<\rho_c+a_{u} x_{c}<\eta_1 \label{eq:etabound}
\ \  \text{and, } \frac{1-\rho_c}{\rho_c-\sigma_c}\notin \QQ^+,
\end{align}
where $\QQ^+$ is set of positive rational numbers.
Note that for both $\sigma_c,\rho_c$, $0<\sigma_c<\rho_c<1$. We define $E^\eta_{cv}(t)$ as the event that $\hat{\eta} _{cv}(t)\leq \sigma_c$, and $E_{c}^\theta(t)$ as the event that
$\theta_c(t)\leq \rho_c+a_u x_c$.  Note that, since $(1-\rho_c)/(\rho_c-\sigma_c)\notin \QQ^+$,
So, we have $d_{KL}((\sigma_c N_{cv}(t)+n_\alpha)/(N_{cv}(t)+n_\alpha),\rho_c)>0$.
Hence, we can define 
\begin{align}
d_m= \min_{c,v}\min_{N\in\NN^+} d_{KL}((\sigma_c N+n_\alpha)/(N+n_\alpha),\rho_c) \ge 0\label{eq:dm}
\end{align}
%
To prove the theorem, in the first step, we show that the regret is proportional to the sum of number of times a suboptimal category $(k_c(T) \text{ for } c\neq 1)$ is posted.
Then we decompose this quantity into three suitable quantities and provide individual bounds and then combine them.
%
 From, definition (xi), we note that $\eta_1=\max_{c}\eta_c$. Therefore, the expected regret is proportional to the number of times a suboptimal category $c\neq 1$ is posted.
\begin{align}
\EE_{\Hcal_T}(\regT)=\sum_{c\in \Ccal}\big(\eta_1-\eta_{c}\big)\EE(k_c(T)) \label{eq:kit}
\end{align}
where we recall that $k_c(T)$ is the number of times a message with category $c$ has been posted up to and including time $T$.
So, $k_c(T)=\sum_{t\in[T]}\mathbf{1}(c_t=c)$. Therefore, we have,
{\small
\begin{align}
&\EE_{\Hcal_T}(k_c(T))=\sum_{t\in[T]} \PP(c_t=c)=\sum_{t\in[T]}\PP\big(c_t=c,\bigcap_{v} E^\eta_{cv}(t), E_{c}^\theta(t)\big)\nn\\
&\hspace*{-0.4cm}+\sum_{t\in[T]}\PP\big(c_t=c,\bigcap_{v} {E^\eta_{cv}(t)},\overline{ E_{c}^\theta (t)}\big)
+\sum_{t\in[T]}\PP\big(c_t=c,\bigcup_{v} \overline{E^\eta_{cv}(t)}\big)\label{eq:kthreeterms}
\end{align}}
Now, we are going to bound these three sums individually.

\noindent \emph{--- Bounding $\PP\big(c_t=c,\bigcap_{v} E^\eta_{cv}(t), E_{c}^\theta(t)\big)$: }
\begin{align}
&\sum_{t\in[T]}\PP(c_t=c,\bigcap_{v} E^\eta_{cv}(t), E_{c}^\theta(t))\nn\\[-0.2cm]
&\overset{a}{\leq}\sum_{t\in[T]} \Et\Big(\frac{1-p_{c,t}}{p_{c,t}}\PP(c_t=1,\bigcap_{v} E^\eta_{cv}(t), E_{c}^\theta(t)) |\Hcal_{t})\Big)\nn \\[-0.2cm]
&\overset{b}{\leq} \sum_{t\in[T]}\EE\Big(\frac{1-p_{c,t}}{p_{c,t}}\mathbf{1}(c_t=1)\Big)
\overset{c}{\leq}\sum_{m=1} ^{T} \EE\Big(\frac{1-p_{c,\tau_m}}{p_{c,\tau_m}}\Big)\nn
\end{align}
%
%
Inequality $a$ comes from Lemma~\ref{lem:ineq1}. Inequality $b$ comes from the fact that $\PP(c_t=1,\bullet) \le \PP(c_t=1)$ and in the corresponding expression, the expectation is taken over all sources of randomness.
Since,  $\tau_m$ denotes the time step at which a message with optimal category (i.e. $c=1$) is posted for the $m$-th time,
at times other than $t=\tau_m$, the indicator term becomes zero, which explains the last inequality $c$.
From Lemma~\ref{lem:ineq2}, we note that $\sum_{m=1}^T \EE\Big(\frac{1-p_{c,\tau_m}}{p_{c,\tau_m}}\Big)$ is
\begin{align}
\sum_{m=1}^T \underset{v\in\Ncal(u)}{\sum} O\Big(e^{-\Delta_{cv} ^2 m/2}+\frac{e^{-D_{cv} m}}{(m+1)\Delta_{cv} ^2}+\frac{1}{e^{\Delta^2 _{cv}m/4}-1}\Big)\label{eq:order1}
\end{align}
which is  $\Theta(1)$.
Next we give bound of the second term in Eq.~\ref{eq:kthreeterms}.\\
\emph{--- Bounding $\sum_{t\in[T]}\PP(c_t=c,\bigcap_{v} {E^\eta_{cv}(t)},\overline{ E_{c}^\theta (t)})$.}
We observe that,
\begin{align}
&\PP(c_t=c,\bigcap_{v} {E^\eta_{cv}(t)},\overline{ E_{c}^\theta (t)})\nn\\[-0.1cm]
&\le \PP(c_t=c,\overline{ E_{c}^\theta (t)}\big|\bigcap_{v} {E^\eta_{cv}(t)},\Hcal_t)\nn\\[-0.1cm]
&\leq \PP\Big(\hspace*{-0.1cm}\sum_{v\in\Ncal(u)}a_v \hat{q}_{cv}(t)>\rho_c| \bigcap_{v\in\Ncal(u)}\{\hat{\eta}_{cv}(t)\leq \sigma_c\},\Hcal_t\Big)\nn\\[-0.1cm]
&\overset{a}{\leq} \PP \Big(\hspace*{-0.1cm} \bigcup_{v\in\Ncal(u)}\{ \hat{q}_{cv}(t)>\rho_c\}| \bigcap_{v\in\Ncal(u)}\{\hat{\eta}_{cv}(t)\leq \sigma_c\},\Hcal_t\Big)\nn\\[-0.1cm]
&\overset{b}{\leq}\sum_{v\in\Ncal(u)} \PP(\hat{q}_{cv}(t)>\rho_c|\hat{\eta}_{cv}(t)\leq \sigma_c)\nn\\[-0.1cm]
&\overset{c}{\leq} \sum_{v\in\Ncal(u)}(1-F^{Beta} _{\sigma_c N_{cv}(t)+n_\alpha+1,(1-\sigma_c) N_{cv}(t)}(\rho_c))\nn\\[-0.1cm]
&\overset{d}{\leq} \sum_{v\in\Ncal(u)}F^B _{N_{cv}(t)+n_\alpha, \rho_c}(\sigma_c N_{cv}(t)+n_\alpha)\nn\\[-0.1cm]
&\overset{e}{\leq}\sum_{v\in\Ncal(u)} \exp\left(-N_{cv}(t)d_{KL}\Big(\frac{\sigma_c N_{cv}(t)+n_\alpha}{N_{cv}(t)+n_\alpha},\rho_c\Big)\right)\label{eq:dklexp}
\end{align}
Inequality $a$ is due to the following. If $\hat{q}_{cv}(t)<\rho_c$ for all $v\in\Ncal(u)$, then $\sum_{v\in\Ncal(u)}a_v \hat{q}_{cv}(t)<\sum_{v\in\Ncal(u)}\rho_c<\rho_c$.
Inequality $c$ is due to the fact that cdf of Beta($\alpha,\beta$) is decreasing w.r.t. $\alpha$ and increasing w.r.t $\beta$~\cite{tspaper}. Inequality $d$ 
is due to the relation between Beta distribution with integer parameters and Binomial distributions~\cite[Fact 3, Appendix A]{tspaper}.
Inequality $e$ is due to Chernoff-Holding bound~\cite[Fact 1,  Appendix A]{tspaper}.
Then, from Eq.~\ref{eq:dklexp}, we obtain that,
\begin{align}
 \PP(c_t=c,\overline{ E_{c}^\theta (t)}\big|\bigcap_{v} {E^\eta_{cv}(t)},\Hcal_t)\leq \sum_{v\in\Ncal(u)}\exp (-N_{cv}(t)d_m)
\end{align}
Now, we split $N_{cv}(t)=M_{cv}(t)+k_c(t-1)$.
For time $t$ such that $k_c(t-1)\geq L_c(T)$, we use the definition of $d_m$ from Eq.~\ref{eq:dm} to have,
\begin{align}
 \PP(c_t=c,\overline{ E_{c}^\theta (t)}\big|\bigcap_{v} {E^\eta_{cv}(t)},\Hcal_t)\leq \sum_{v\in\Ncal(u)} \exp(-d_m(M_{cv}(t)+L_c(T)))\nn
\end{align}
Now consider that $t_s$ is the largest time until $k_c(t-1)\leq L_c(T)$.
Then we have,
\begin{align}
 &\sum_{t\in[T]}\PP(c_t=c,\bigcap_{v} {E^\eta_{cv}(t)},\overline{ E_{c}^\theta (t)})\nn\\[-0.1cm]
 &\leq \sum_{t\in[T]}\Et\Big(\PP(c_t=c,\overline{ E_{c}^\theta (t)}\big|\bigcap_{v} {E^\eta_{cv}(t)},\Hcal_t)\Big)\nn\\[-0.1cm]
 &\leq \Et\Big(\sum_{t\in[t_s]} \PP(c_t=c,\overline{ E_{c}^\theta (t)}\big|\bigcap_{v} {E^\eta_{cv}(t)},\Hcal_t)\Big)\nn\\[-0.1cm]
 &+\Et\Big(\sum_{t=t_s+1}^T \PP(c_t=c,\overline{ E_{c}^\theta (t)}\big|\bigcap_{v} {E^\eta_{cv}(t)},\Hcal_t)\Big)\nn\\[-0.1cm]
 &\leq \EE\big(\sum_{t\in[t_s]} \mathbf{1}(c_t=c)\big)\nn
  \end{align}
 \begin{align}
 &+\sum_{v\in\Ncal(u)}\sum_{t=t_s+1}^T \EE_{M_{cv}(t)}\exp(-d_mM_{cv}(t))\exp(-d_mL_c(T))\nn\\[-0.1cm]
 &\leq L_c(T)+\hspace*{-0.3cm}\sum_{v\in\Ncal(u)}\int_0 ^T \exp(-\mu_{cv} (1-e^{-d_m})t) dt \exp(-d_mL_c(T)) \nn \\[-0.1cm]
& \leq L_c(T)+\hspace*{-0.2cm}\sum_{v\in\Ncal(u)}\hspace*{-0.1cm}\frac{1-\exp(-\mu_{cv} (1-e^{-d_m}) T)}{\mu_{cv} (1-e^{-d_m})} \exp(-d_mL_c(T)).\nn \vspace*{-1cm}
 \end{align}
 
 Then we have, $\sum_{t\in[T]}\PP(c_t=c,\bigcap_{v} {E^\eta_{cv}(t)},\overline{ E_{c}^\theta (t)})$ has order
 \vspace*{-0.4cm}
 \begin{align}
  O\left(\log\Big(1+\sum_{v\in\Ncal(u)}\frac{1-\exp(-\mu_{cv} \theta_m T)}{\mu_{cv} \theta_m }\Big)\right)\label{eq:order2}
 \end{align}
 \vspace*{-0.1cm}
 where, $\theta_m=1-e^{-d_m}$.\\
 \emph{--- Bounding $\sum_{t\in[T]}\PP(c_t=c,\bigcup_{v} \overline{E^\eta_{cv}(t)} )$.}
 We define $t_k$ is the time at which category $c$ is posted $k^{th}$ time. Then, we observe that
 \begin{align}
  &\sum_{t\in[T]}\PP(c_t=c,\bigcup_{v} \overline{E^\eta_{cv}(t)} )\nn
   \leq \EE \big[\sum_{t\in[T]}\mathbf{1}(c_t=c)\mathbf{1}(\cup_{v} \overline{E^\eta_{cv}(t)})  \big]\nn\\[-0.1cm]
 & \overset{a}{\leq} 1+\EE \big[\sum_{m\in[T]}\mathbf{1}(\cup_{v} \overline{E^\eta_{cv}(t_m)})  \big]
 \overset{b}{\leq} 1+\sum_v \EE \big[\sum_{m\in[T]}\mathbf{1}(\overline{E^\eta_{cv}(t_m)})\big]\nn\\[-0.1cm]
 &\overset{c}{\leq}1+\sum_v\sum_{m\in[T]}\exp(-m d_{KL}(\sigma_c,q_{cv}))=O(1)\label{eq:order3}
 \end{align}
 Inequality $a$ holds because the indicator $\mathbf{1}(\overline{E^\eta_{cv}(t_k)})$ is only activated when
 category $c$ posted. Inequality $b$ is due to an union bound and inequality $c$ is due to Chernoff-Hoeffding bounds ~\cite[Fact 1, Appendix A]{tspaper}.
Combining Eq.~\ref{eq:order1}, ~\ref{eq:order2}, ~\ref{eq:order3}, we obtain the required bound.
%
\vspace*{-0.1cm}
\subsection{Auxilliary lemmas}
%
\begin{lemma}
If the rate of messages in $v$'s feed are exposed to $u$ with Poisson distribution having rate $\mu_{cv} $ for each category $c$, and $M_{cv}(t)$ are total
number of such messages for category $c$, observed until timestep $t$, then we have
\vspace*{-0.4cm}
\begin{align}
 \EE\Big(\frac{1}{M_{cv}(t)+\delta}\Big)\le
 \begin{cases}
 \frac{1-e^{-\mu_{cv}(t-1)}}{(t-1)\mu_{cv}\delta}\ \text{if }\delta<1\\
  \frac{1-e^{-((t-1)\mu_{cv}+\delta-1)}}{(t-1)\mu_{cv}+\delta-1}\ \text{if }\delta\geq 1
 \end{cases}
 \end{align}
\end{lemma}
\textbf{Proof:}\
 $ \EE\Big(\frac{1}{M_{cv}(t)+\delta}\Big)$ can be written as,
 \begin{align}
  &\sum_{n=0}^\infty\frac{1}{n+\delta}\frac{((t-1)\mu_{cv})^n e^{-(t-1)\mu_{cv}}}{n!}\\
   &=\int_0 ^1 \sum_{n=0}^\infty \frac{(\mu_{cv}(t-1)\theta)^n e^{-(t-1)\mu_{cv}}}{n!} \theta^{\delta-1}d\theta \nn \\
   &=\int_0 ^1 e^{-\mu_{cv}(t-1)(1-\theta)}\theta^{\delta-1}d\theta\label{int1}
 \end{align}
\emph{--- Case  $\delta<1$.}  In this case, $e^{-\mu_{cv}(t-1)(1-\theta)}$ is increasing and $\theta^{\delta-1}$ is decreasing. So we apply Chebyschev inequality~\cite{ostrowski2010integral} to have
\begin{align}
 &\int_0 ^1 e^{-\mu_{cv}(t-1)(1-\theta)}\theta^{\delta-1}d\theta
 \le \int_0 ^1 e^{-\mu_{cv}(t-1)(1-\theta)}d\theta\int_0 ^1 \theta^{\delta-1}d\theta\nn
\end{align}0.4
which provides the required expression.\\
 \emph{--- Case  $\delta\ge 1$}. Since $\theta\le e^{-(1-\theta)}$, $\delta\geq 1$ implies $\theta^{\delta-1} \leq e^{-(1-\theta)(\delta-1)}$. Hence, the required integral is less than
\begin{align}
 \int_0 ^1 e^{-((t-1)\mu_{cv}+\delta-1)(1-\theta)}d\theta =\frac{1-e^{(t-1)\mu_{cv}+\delta-1}}{(t-1)\mu_{cv}+\delta-1}\nn
\end{align}
 \begin{lemma}\label{lem:key2}
If $M_{cv}(t)$ and $\mu_{cv}$ have the similar meanings as the previous lemma, then we have
\begin{align}
\sum_{t=2} ^T  \sqrt{\EE\Big(\frac{1}{M_{cv}(t)+\delta}\Big)}\le
 \begin{cases}
\frac{\sqrt{T-1}\sqrt{1-e^{-(\mu_{cv} (T-1) )}}}{2\sqrt{\mu_{cv}\Delta  t \delta}} \  \text{if }\delta<1\\
\frac{T-1}{2\sqrt{(T-1)\mu_{cv}+\delta-1}+\sqrt{\delta-1}}\  \text{if }\delta\ge 1\\
 \end{cases}
 \end{align}
\end{lemma}
\textbf{Proof:}\
For $\delta <1$, we have
\begin{align}
 \sum_{t=2} ^T \sqrt{\EE\Big(\frac{1}{M_{cv}(t)+\delta}\Big)}&\le
\int _{1} ^T  \sqrt{ \frac{1-e^{-\mu_{cv}(T-1)}}{(t-1)\mu_{cv}\delta}} dt,\nn
\end{align}
which gives the required bound. For $\delta \ge 1$, we have
\begin{align}
 &\sum_{t=2} ^T  \sqrt{\EE\Big(\frac{1}{M_{cv}(t)+\delta}\Big)} \le
\int _{1} ^T \frac{1}{\sqrt{\mu_{cv} (t-1)+\delta -1}} dt\nn\\
&=\frac{\sqrt{(T-1)\mu_{cv}+\delta-1}-\sqrt{\delta-1}}{2\mu_{cv}}  =\frac{T-1}{2\sqrt{(T-1)\mu_{cv}+\delta-1}+\sqrt{\delta-1}} \nn
\end{align}
 %
  \begin{lemma}~\cite{tspaper}\label{lem:ineq1}
  If we define $p_{c,t}=\PP(\theta_1(t)>\rho_c+a_u x_c|\Hcal_{t})$, then 
  \begin{align}
  \hspace*{-0.3cm}&\PP(c_t=c,\bigcap_{v} E^\eta_{cv}(t), E_{c}^\theta(t)) |\Hcal_{t}) \leq \frac{1-p_{c,t}}{p_{c,t}}\PP(c_t=1,\bigcap_{v} E^\eta_{cv}(t), E_{c}^\theta(t)) |\Hcal_{t})\nn
  \end{align}
  \end{lemma}

  \begin{lemma}\label{lem:ineq2} 
  If $\tau_j$ denote the time step at which the optimal category $1$ is selected for the $j$-th time.
  Then,
  \begin{align}
  \hspace*{-0.3cm} \EE\Big(\frac{1}{p_{c,\tau_j}}\Big)=
    1+\hspace*{-0.3cm}\underset{v\in\Ncal(u)}{\sum}\hspace*{-0.1cm}O\big(e^{-\Delta_{cv} ^2 j/2}+\frac{e^{-D_{cv} j}}{(j+1)\Delta_{cv} ^2}+\frac{1}{e^{\Delta^2 _{cv} j/4}-1}\big),\nn
  \end{align}
  where, $\Delta_{cv}=q_{1v}-y_{cv}$, $D_{cv}=y_{cv}\log(y_{cv}/q_{1v})+(1-y_{cv})\log((1-y_{cv})/(1-q_{1v}))$,
  $\Delta_c=\min_{v\in\Ncal(u)}\Delta_{cv}$; $y_{cv}=\frac{q_{1v}(\rho_c+a_u (x_c-x_1))}{\sum_{w\in\Ncal(u)}a_w q_{1w}}$.
  Note that from Eq.~\ref{eq:etabound},  $\Delta_{cv}>  0$ $\forall v\in\Ncal(u)$.
  \end{lemma}
\begin{proof}
We leverage the proof of~\cite[Lemma 4]{tspaper} to prove this lemma
Let $s_v=n_{1v}(\tau_j)$, $Y_{vj}=N_{1v}(\tau_j)$.
\vspace*{-0.3cm}
\begin{align}
p_{c,\tau_j}=&\PP(\theta_1(\tau_j)>\rho_c+a_u x_c|\Hcal_\tau)
\overset{a}{\geq} \PP\big(\hspace*{-0.2cm}\bigcap_{v\in\Ncal(u)}\hspace*{-0.2cm}\{\hat{q}_{1v}(\tau_j)>y_{cv}\}\big)\nn\\[-0.2cm]
&=\prod_{v\in\Ncal(u)}\PP(\hat{q}_{1v}(\tau_j)>y_{cv})
\end{align}
Inequality $a$ is because $\hat{q}_{1v}(\tau_j)>y_{cv}$ for all $v$ implies
\begin{align}
&\sum_{v\in\Ncal(u)} a_v \hat{q}_{1v}(\tau_j) > \sum_{v\in\Ncal(u)} a_v y_{cv}=  \rho_c+a_u(x_c-x_1),\nn\\[-0.1cm]
&\text{so, }\underbrace{\sum_{v\in\Ncal(u)} a_v \hat{q}_{1v}(\tau_j)+a_u x_1}_{\theta_1(\tau_j)}>\rho_c+a_u x_c.\nn
\end{align}
\vspace*{-0.6cm}
\begin{align}
\text{Now, }\PP(\hat{q}_{1v}(\tau_j)>y_{cv})&=1-F^{Beta} _{s_v+\alpha,Y_{vj}-s_v+\beta}(y_{cv})\nn\\[-0.05cm]
& \overset{a}{\geq}  1-F^{Beta} _{s_v+1,Y_{vj}-s_v+n_\beta}(y_{cv}), \ \nn\\
& \overset{b}{=} F^B _{Y_{vj}+n_\beta,y_{cv}}(s_v)\nn
\end{align}
Inequality $a$ is due to the fact that $F^{Beta} _{\alpha,\beta}(y)$ is increasing in $\beta$ and decreasing in $\alpha$, and we choose $n_\beta>max(\beta,8/\Delta_{cv})$.
Equality $b$ holds due to for simple mapping from Beta distribution to Binomial distribution when the parameters for Beta distribution
are integers. Then, we have:
\begin{align}
\EE(1/{p_{c,\tau_j}}| Y_{vj}) &=\sum_{s_v=0}^{Y_{vj}}  \prod_{v\in\Ncal(u)}  {f_{Y_{vj} ,q_{1v}}(s_v)}   \prod_{v\in\Ncal(u)} \frac{1}{F^B _{Y_{vj}+n_\beta,y_{cv}}(s_v)}\nn\\
&\prod_{v\in\Ncal(u)}\sum_{s_v=0}^{Y_{vj}}\frac{f_{Y_{vj} ,q_{1v}}(s_v)}{F^B _{Y_{vj}+n_\beta,y_{cv}}(s_v)}\label{eq:prodallpt}
\end{align}
We denote $Sum(0,Y_{vj})=\sum_{s_v=0}^{Y_{vj}}\frac{f_{Y_{vj} ,q_{1v}}(s_v)}{F^B _{Y_{vj}+n_\beta,y_{cv}}(s_v)}$, and observe that
$Sum(0, Y_{vj})=Sum(0, Y_{vj}+n_{\beta}-1)$. This is because, $f_{Y_{vj},q_{1v}}(s)=0$ for $s> N_{cv}$.
We express $Sum(0, Y_{vj}+n_{\beta}-1)$  as the sum of two terms
$Sum(0,\lfloor (Y_{vj}+n_{\beta}-1)\delta \rfloor )+Sum(\lceil (Y_{vj}+n_{\beta}-1)\delta \rceil ,Y_{vj}+n_\beta-1)$,  where $\delta=(q_{1v}-\Delta_{cv}/2)$. We note that,
\begin{align}
 &Sum(0,\lfloor (Y_{vj}+n_{\beta}-1)\delta \rfloor) =  \hspace*{-0.4cm} \sum_{s_v=0}^{\lfloor (Y_{vj}+n_{\beta}-1)\delta \rfloor} \hspace*{-0.1cm}\frac{f_{Y_{vj} ,q_{1v}}(s_v)}{F^B _{Y_{vj}+n_\beta,y_{cv}}(s_v)} \nn\\
 &\overset{a}{\leq} \frac{1}{(1-q_{1v})^{n_\beta-1}}\sum_{s_v=0}^{\lfloor (Y_{vj}+n_\beta-1)\delta \rfloor}\frac{f_{Y_{vj}+n_\beta-1,q_{1v}}(s_v)}{F^B _{Y_{vj}+n_\beta,y_{cv}}(s_v)}\nn\\
  &\overset{b}{\leq} \Theta(e^{-\Delta_{cv} ^2 (Y_{vj}+{n_\beta}-1)}+\frac{e^{-D_{jv} (Y_{vj}+n_\beta-1)}}{((Y_{vj}+n_\beta))\Delta_{cv} ^2})\nn\\
& {\leq} \Theta(e^{-\Delta_{cv} ^2 j}+\frac{e^{-D_{jv} j}}{(j+1)\Delta_{cv} ^2}) \ \ \text{(Since, $N_{jv}\geq j $)}\label{eq:firstsum}
  \end{align}
\noindent Inequality $a$ is due the fact that
\begin{align}
 f_{Y_{vj}+n_\beta-1,q_{1v}}(s_v)&= {Y_{vj}+n_\beta -1 \choose s_v} q_{1v} ^{s_v} (1-q_{1v})^{Y_{vj}+n_\beta-1-s_v}\nn\\
 &\geq (1-q_{1v})^{n_\beta-1} {Y_{vj} \choose s_v} q_{1v} ^{s_v} (1-q_{1v})^{Y_{vj}-s_v}.\nn
\end{align}
Inequality $b$ is obtained by leveraging the proof of Lemma 4 in Agarwal \emph{et al} in~\cite{tspaper}.
Now we bound $Sum(\lceil (Y_{vj}+n_{\beta}-1)\delta \rceil ,Y_{vj})$. Using the proof of Lemma 4 in  Agarwal \emph{et al} in~\cite{tspaper},
we have, $F^B _{Y_{vj}+n_\beta,y_{cv}}(s_v)> 1-e^{-\Delta_{cv} ^2 (Y_{vj}+n_\beta-1)/4 }$ for $s_v\geq \lceil (Y_{vj}+n_\beta-1)\delta \rceil$.
Then, we have
\begin{align}
 &Sum(\lceil (Y_{vj}+n_{\beta}-1)\delta \rceil ,Y_{vj}+n_\beta-1)\nn\\
 &= \sum_{s_v=\lceil (Y_{vj}+n_{\beta}-1)\delta \rceil}^{Y_{vj}+n_\beta-1} \hspace*{-0.1cm}\frac{f_{Y_{vj} ,q_{1v}}(s_v)}{F^B _{Y_{vj}+n_\beta,y_{cv}}(s_v)} \nn\\
& \leq 1/(1-e^{-\Delta_{cv} ^2 (Y_{vj}+n_\beta-1)/4 })
  \leq 1/(1-e^{-\Delta_{cv} ^2 j/4 })\nn\\
  &=1+\frac{1}{e^{\Delta_{cv} ^2 j/4 }-1}\label{eq:secondsum}
\end{align}
Adding, Eq.~\ref{eq:firstsum} and~\ref{eq:secondsum}, and then substituting back to Eq.~\ref{eq:prodallpt}, we observe that,
$\EE(1/{p_{c,\tau_j}})$ equals to
\begin{align}
 \prod_{v\in\Ncal(u)}\left(1+\Theta\big(e^{-\Delta_{cv} ^2 j/2}+\frac{e^{-D_{cv} j}}{(j+1)\Delta_{cv} ^2}+\frac{1}{e^{\Delta^2 _{cv} j/4}-1}\right)\nn\\
  \leq 1+\hspace*{-0.3cm}\underset{v\in\Ncal(u)}{\sum}\hspace*{-0.1cm}O\big(e^{-\Delta_{cv} ^2 j/2}+\frac{e^{-D_{cv} j}}{(j+1)\Delta_{cv} ^2}+\frac{1}{e^{\Delta^2 _{cv} j/4}-1}\big),\nn
\end{align}
The last inequality follows by taking only the dominating term.

\end{proof}